\documentclass[12pt]{article}
\usepackage{setspace}
\usepackage{amsfonts}%this package allows for calligraphic letters
\usepackage{comment}
\usepackage{longtable}
\usepackage{pdflscape}
\usepackage{tikz}
\usepackage{bibentry}
\usepackage{mathrsfs}
\nobibliography*
\usepackage[margin=1in]{geometry}
\usepackage{enumerate}
\usepackage{amssymb,amsmath,amsthm}
\usepackage{bbm}
\usepackage{array}
\usepackage{eqnarray}
\usepackage{tabularx}
\usepackage{graphicx}
\usepackage{subcaption} %this package allows side by side figures
\usepackage[authoryear]{natbib}
\usepackage{booktabs}
\usepackage{titlesec}
\usepackage{cite}
\usepackage{rotating}
\usepackage{soul}
\usepackage{fancyhdr}
\usepackage{multirow}
\usepackage{threeparttable}
\usepackage{cancel}
\usepackage{xr}
\usepackage{titletoc}
\usepackage{float}
\usepackage{tikz}
\usetikzlibrary{decorations.pathreplacing}
\usepackage[hypertexnames=false]{hyperref}
\hypersetup{
    colorlinks,
    linkcolor={blue!30!black},
    citecolor={blue!30!black},
    urlcolor={blue!30!black},
    breaklinks=true
}
\usepackage{xurl}
\titlespacing*{\section}{0pt}{0.5ex plus 1ex minus .2ex}{0.3ex plus .2ex}
\titlespacing*{\subsection}{0pt}{0.25ex plus 1ex minus .2ex}{0.3ex plus .2ex}
\titlespacing*{\subsubsection}
{0pt}{0.25ex plus 1ex minus .2ex}{0.3ex plus .2ex}
\usepackage{array}
\newcolumntype{L}[1]{>{\raggedright\let\newline\\\arraybackslash\hspace{0pt}}m{#1}}
\newcolumntype{C}[1]{>{\centering\let\newline\\\arraybackslash\hspace{0pt}}m{#1}}
\newcolumntype{R}[1]{>{\raggedleft\let\newline\\\arraybackslash\hspace{0pt}}m{#1}}

\usepackage{fancyhdr} % This should be set AFTER setting up the page geometry
\usepackage{sectsty}
\sectionfont{\large}
\subsectionfont{\normalsize}
\subsubsectionfont{\normalsize}

\usepackage{xcolor}
\definecolor{forest}{rgb}{0.0, 0.5, 0.0}

\usepackage[normalem]{ulem}
\graphicspath{{./Figures/}}

\newtheorem{proposition}{Proposition}[section]
\newtheorem{corollary}{Corollary}[section]

\newtheorem{remark}{Remark}[section]

\usetikzlibrary{arrows, positioning}
\usetikzlibrary{calc}

 \title{To what extent can long-differencing capture climate adaptation?\thanks{We thank participants at the 2026 World Congress of Environmental and Resource Economists for helpful comments and suggestions. Ghanem gratefully acknowledges the support of National Institute of Food and Agriculture grant 2024-67023-42704. We used Gemini (3 and 3.5) and Claude (Sonnet 5) to assist with coding and Claude (Fable 5) to assist with algebraic derivations for the analytical examples.}}
  \author{Dalia Ghanem\thanks{Department of Agricultural and Resource Economics, University of California Davis, \url{dghanem@ucdavis.edu}.}\quad Felix Pretis\thanks{Department of Economics, University of Victoria, British Columbia, \url{fpretis@uvic.ca}.}\quad Daniel Schuurman\thanks{Department of Agricultural and Resource Economics, University of California Davis, \url{dschuurman@ucdavis.edu}.}}
  \date{First Draft: July 24, 2026\\Comments welcome!}
\begin{document}
\maketitle

\begin{abstract}
Understanding the degree to which we are able to adapt to climate change is central to economic assessments of future climate damages. Economists increasingly use comparisons between long differences and fixed effects estimators to measure climate adaptation. We show that such comparisons can be misleading. Neither estimator is consistent for its intended parameter, as both the long-difference (LD) and fixed effects (FE) estimands are weighted averages of the long- and short-run responses to climate and weather. As a result, the difference between the two understates the true extent of adaptation, and the standard test based on this difference --while controlling size -- tends to be substantially underpowered in the settings researchers typically encounter. An empirically-calibrated simulation shows this difference understates adaptation by about 30--80\%, depending on the averaging window.
\end{abstract}

\newpage

\section{Introduction}

The economic costs of future climate change will depend not only on the trajectory of warming itself, but on the extent to which households, firms, and governments adapt to the changing climate they face. However, there is substantial uncertainty around the degree of existing adaptation and what the limits to adaptation might be. Economists are thus increasingly interested in measuring the extent of adaptation to climate change \citep{carleton2024hecc,lemoine2025nber,kolstad2020reep}. A common approach to estimating and testing for adaptation uses long differences (LD) \citep{dell2012aejm,burke2016aejep} to isolate climate variation from transitory weather shocks, comparing multi-year averages to filter out the transitory shocks and reveal how systems adjust to permanent environmental shifts over time. Adaptation is then inferred by comparing LD estimates of long-run responses to climatic variables with standard panel fixed effects (FE) estimates of short-run responses to observed weather. This empirical approach has been used to study adaptation and the long-run effects of climate in a wide range of contexts, including agriculture \citep{won2024ajae,cui2025,chen2021jde,yu2021sr,taylor2026jaere}, labor  \citep{liu2023aejep}, manufacturing \citep{ponticelli2023nber}, migration \citep{obolensky2024nber,baylis2025jpub}, and health \citep{obradovich2018pnas,carleton2017pnas}. Despite the wide use of LD to FE comparisons, little is known about whether this comparison actually recovers adaptation. The goal of this paper is to formally examine the conditions under which comparisons of FE and LD estimators can be informative about the extent of adaptation. Our results show that neither LD nor FE estimators are consistent for their respective target parameters, and thus resulting tests for adaptation can be substantively underpowered.

We first show that the LD and FE estimands are weighted averages of the long- and short-run response. Intuitively, the weight on the short-run response in the LD estimand depends on the contribution of short-run weather variation that survives the LD transformation. We therefore refer to it as the LD ``contamination weight''. The weight on the long-run response in the FE estimand depends on the contribution of climate variation that survives the FE transformation. We refer to it as the FE contamination weight. Under mild conditions, both LD and FE estimands are attenuated towards zero and away from their respective targets. 

This formal analysis relies on a standard outcome model used in the literature \citep[e.g.][]{burke2016aejep} to justify long-differencing, but does not require any assumptions on how the climate process itself evolves. Indeed, an important takeaway from our results is that long-differencing is agnostic about the climate process, though not about how the outcome of interest is generated given climate and weather shocks.\footnote{We show how the probability limits are affected by a deviation from the proposed outcome model to allow for response heterogeneity (see Remark \ref{rem:heterogeneity}).}

Using the probability limits, we can show that comparing FE and LD estimands constitutes a test by implication of the null hypothesis of no adaptation. The equality of the FE and LD estimands is a necessary but not sufficient condition of the no-adaptation hypothesis. If the LD and FE contamination weights sum to one, then this test would have trivial power regardless of the extent of adaptation. Caution is therefore warranted when interpreting non-rejections of adaptation tests, as with other tests that are based on an implication of the null hypothesis in question (such as identification tests), as opposed to an equivalent condition.

We illustrate the formal analysis numerically using an empirically-calibrated simulation using the temperature dataset used in \citet{burke2016aejep}. Consistent with our formal results, this simulation design demonstrates that the difference between LD and FE estimators is (in most conventional settings) downwardly biased relative to the true extent of adaptation. This downward bias can range from 30-80\% depending on the choice of averaging window used in the LD estimator and can lead to substantive loss in finite-sample power. We also report simulation results for the contamination weights, which underscore that the LD contamination weight tends to be larger in magnitude relative to the FE contamination weight. Furthermore, the LD contamination weight can vary substantially with the choice of the averaging window in the LD estimator.

In order to formally examine the role of the averaging window in the LD contamination weight, we consider two analytical examples as well as numerical examples using temperature data. We analyze the contamination weights analytically for two time series processes, specifically a linearly trending and a unit root case. For both examples, our analysis demonstrates that the relationship between the averaging window and the LD contamination weight is nonlinear and depends on the climate process. We demonstrate the corresponding patterns using the temperature data used in our simulations, where the climate component is defined as the 10- or 30-year normal. This analysis demonstrates that the specific choice of the bias-minimizing averaging window depends on the unobserved climate process and the LD contamination weight can be substantive, even at the minimizing choice of averaging window. In order to assess the magnitude of this weight in practice, researchers would have to consider different climate specifications. The contamination weights estimated using these specifications can also be used to construct confidence intervals for the true extent of adaptation by test inversion. We discuss these implications for empirical practice in Section \ref{sec:implications}.

This paper highlights econometric issues with the use of long differencing to measure adaptation. Concerns about the variation the LD method uses to identify long-run response were illustrated in \citet{burke2016aejep} and discussed in other work \citep{kolstad2020reep,lemoine2018nber,lemoine2025nber}. We demonstrate that LD and FE are biased away from their respective targets and demonstrate the implications for testing and measuring the extent of adaptation based on their difference.

Our work also contributes to a broad body of methodological work measuring long-term climate impacts and adaptation. Among the methods currently used, including non-linear panel data methods \citep{merel2021ajae,deryugina2017nber}, multistage models \citep{auffhammer2022jeem,butler2013ncc}, and partitioning variation approaches \citep{bento2023jeem,bilal2024nber,moore2014ncc,gammans2017}, LD is a commonly used approach. Our formal analysis provides guidance for interpreting adaptation tests and measures based on differences between LD and FE.

Our analysis is closely related to work on measurement error in panel models (\citealt{griliches1986errors}) where within- and first-difference transformations split regressor variation in distinct ways and therefore produce predictable patterns of bias. In their classical errors-in-variables setting, \citet{griliches1986errors} decompose observed regressors into a true signal and an i.i.d.\ measurement error. They compare estimators across within and first-difference transformations to characterize attenuation bias and recover the true coefficient. Our setting shares the similar feature of a regressor decomposed into two components with different time-series properties, but our analysis departs in a substantive way: the high-frequency component in our setting is not measurement error but instead a weather shock carrying its own causal short-run response. Differential filtering therefore does not produce attenuation toward zero, but yields estimators that are weighted averages of long- and short-run responses, with biases of opposite signs relative to their respective targets. Whereas \citeauthor{griliches1986errors} use differential filtering to recover a single coefficient, our analysis shows that a test for climate adaptation attenuates the very difference between long-run and short-run responses that the test is meant to detect. Last but not least, we do not (and cannot) assume that climate and weather shocks are uncorrelated as in the classical measurement error problem, as they are necessarily dependent.

\section{FE and LD estimators in the presence of adaptation}
\label{sec:theory_section}
In this section, we demonstrate that the FE and LD estimands are biased away from their respective targets, characterize their contamination weights and show that the adaptation test is a test by implication, highlighting its dependence on the bias of both estimands. We then illustrate this analysis in an empirically-calibrated simulation study.
\subsection{Setup}
FE models are used widely to estimate the effect of weather on economic outcomes of interest. The predominant specifications examined in this literature are linear in the parameters, while nonlinear in the higher-frequency (e.g.\ daily) temperature
\citep[for a review of the econometric specifications, see][]{cui2024on}
\begin{eqnarray}
y_{it}&=&\beta x_{it}+a_i+\lambda_t+z_{it}'\gamma+\varepsilon_{it}\quad i = 1,...,n, \quad t = 1,...,T. \label{eq:y_x_model}\\
    x_{it}&=&\mu(\mathcal{W}_{it})\nonumber
\end{eqnarray}
where $i$ denotes the cross-sectional unit, $t$ denotes the time period (commonly year), and $h$ the higher frequency dimension, which is typically daily. The higher-frequency temperature time series is denoted by $\mathcal{W}_{it}\equiv \{W_{it1},\dots,W_{itH}\}$. To simplify notation, we assume that $x_{it}$ is scalar, but the framework allows $x_{it}$ to be multi-dimensional and include bins, degree day measures, splines and other popular specification choices as a special case.

To allow cross-sectional units to respond differently to short- and long-run changes in weather, suppose that observed weather $x_{it}$ is the sum of climate $c^*_{it}$ and weather shocks $w^*_{it}$ that are both latent to the econometrician, $x_{it}=c^*_{it}+w^*_{it}$.\footnote{For instance, if climate only changes the mean as assumed in some climate models \citep{burke2016aejep,deryugina2017nber}, such that $x_{it}\sim (\mu_{it},\sigma_x^2)$, then climate is given by $c_{it}^*=\mu_{it}$ and weather by deviations from climate $w_{it}^*=x_{it}-\mu_{it}\sim (0,\sigma_w^2)$. The assumption that weather is composed of a low-frequency component representing climate and a high-frequency component representing weather shocks is also present in work including \citet{gospodinov2025}, \citet{bilal2024nber}, and \citet{bento2023jeem}. } The following specification \citep[e.g.][]{burke2016aejep} allows $\theta_S$ to capture the short-run response to $w_{it}^*$ and $\theta_L$ to capture the long-run response to $c_{it}^*$,\footnote{See \citet{burke2016aejep} Equation (2) in Supplemental Appendix Section A.2.1.} 
\begin{equation}
  y_{it} = \theta_S w^*_{it} + \theta_L  c^*_{it} +\alpha_i+ \lambda_t+z_{it}'\gamma+\varepsilon_{it}, \quad i = 1,...,n, \quad t = 1,...,T. \label{eq:outcome_model_covariates}
\end{equation}
In the absence of adaptation ($\theta_L=\theta_S$), Eq.\ \eqref{eq:outcome_model_covariates} simplifies to Eq.\ \eqref{eq:y_x_model}; that is, the latter is a restricted version of the former.

The setup in Eq.\ \eqref{eq:outcome_model} parallels that of \citet{griliches1986errors}, who study within- and first-difference estimators under classical measurement error. The key distinction is that they treat the second component as a nuisance to be eliminated, whereas we treat it as a weather shock with its own causal effect — which, as shown below, changes the nature of the resulting bias.

To simplify presentation hereinafter, we abstract from covariates and additional fixed effects in Eq.\ \eqref{eq:y_x_model} and consider the following model,
\begin{eqnarray}y_{it}&=&\theta_S w_{it}^*+\theta_L c_{it}^*+\alpha_i+\varepsilon_{it}\label{eq:outcome_model}\end{eqnarray}
We emphasize however that this omission is without loss of generality as the extension of our results to \eqref{eq:outcome_model_covariates} is immediate by Frisch-Waugh-Lovell Theorem.

\begin{remark}[Beyond homogeneous response models] 
We extend our results to heterogeneous response models in Appendix \ref{app:heterogeneity} and summarize the main takeaways in Remark \ref{rem:heterogeneity}.
\end{remark}
\subsection{Definitions: FE, LD, and adaptation test}
In this section, we introduce the FE and LD estimators as well as the adaptation test.

Let $\tilde z_{it}$ represent the within-transformed version of $z_{it}$, formally $\tilde{z}_{it}\equiv z_{it}- \frac{1}{T}\sum_{t=1}^Tz_{it}$. To simplify notation, we use $\sum_t$ to denote $\sum_{t=1}^T$. The FE model is given by the following,
\begin{equation}
    \tilde y_{it} = \beta_{FE}\tilde x_{it} + \tilde u_{it}.
\end{equation}
The population analogue of the FE estimator, $\hat{\beta}_{FE}$, as $n\rightarrow\infty$ is denoted by $\beta_{FE}\equiv \dfrac{\bar{E}[\sum_t\tilde{x}_{it}\tilde{y}_{it}]}{\bar{E}[\sum_t\tilde{x}_{it}^2]}$, where $\bar{E}[Z_i]\equiv \lim_{n\rightarrow\infty}\frac{1}{n}\sum_{i=1}^nE[Z_i]$ and we assume sufficient regularity conditions hold.
 
The LD approach relies on a different transformation of the outcome and regressor that takes the difference of their average in two equal-sized, non-overlapping windows. The logic is that temporal averaging cancels out transitory weather fluctuations, leaving only climate variation in the differences between averages. The LD-transformed model is given by
\begin{equation}
   \Delta \bar y_{i} = \beta_{LD}\Delta \bar x_{i} + \Delta \bar u_{i},
\end{equation}
where $\Delta \bar z_{i}$ is the long-differenced version of $z_{it}$, formally $\Delta \bar z_{i}\equiv \frac{1}{\tau}\sum_{t=T-\tau+1}^T z_{it}-\frac{1}{\tau}\sum_{t=1}^\tau z_{it}$ for some $1\leq \tau\leq T/2$. Let $\beta_{LD}$ denote the population analogue of the LD estimator $\hat{\beta}_{LD}$, formally $\beta_{LD}\equiv \frac{\bar{E}[\Delta \bar x_i\Delta \bar y_i]}{\bar{E}[\Delta \bar x_i^2]}$, where we assume sufficient moment conditions.\footnote{With a slight abuse of notation, we will use $\Delta \bar z_i^2$ to denote the square of $(\Delta \bar z_i)^2$ to avoid having too many parantheses throughout.} The LD estimand therefore depends on the user's choice of $\tau$: the time span over which the average is computed in the LD transformation. Figure \ref{fig:ld_lit} presents the choices of $\tau$ together with the study period ($T$) in recent articles that have used LD and demonstrates a lack of consensus on the choice of this tuning parameter.

\begin{figure}[htbp]
    \centering
    \begin{tabular}{cc}
    \includegraphics[width=0.45\linewidth]{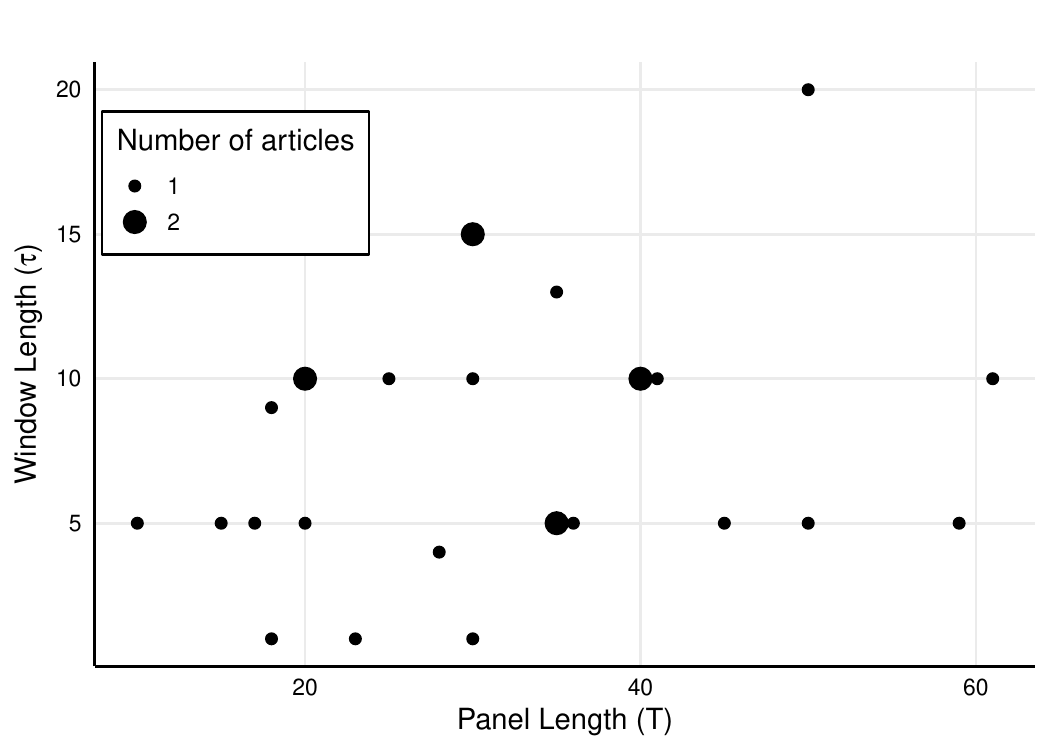}&  \includegraphics[width=0.45\linewidth]{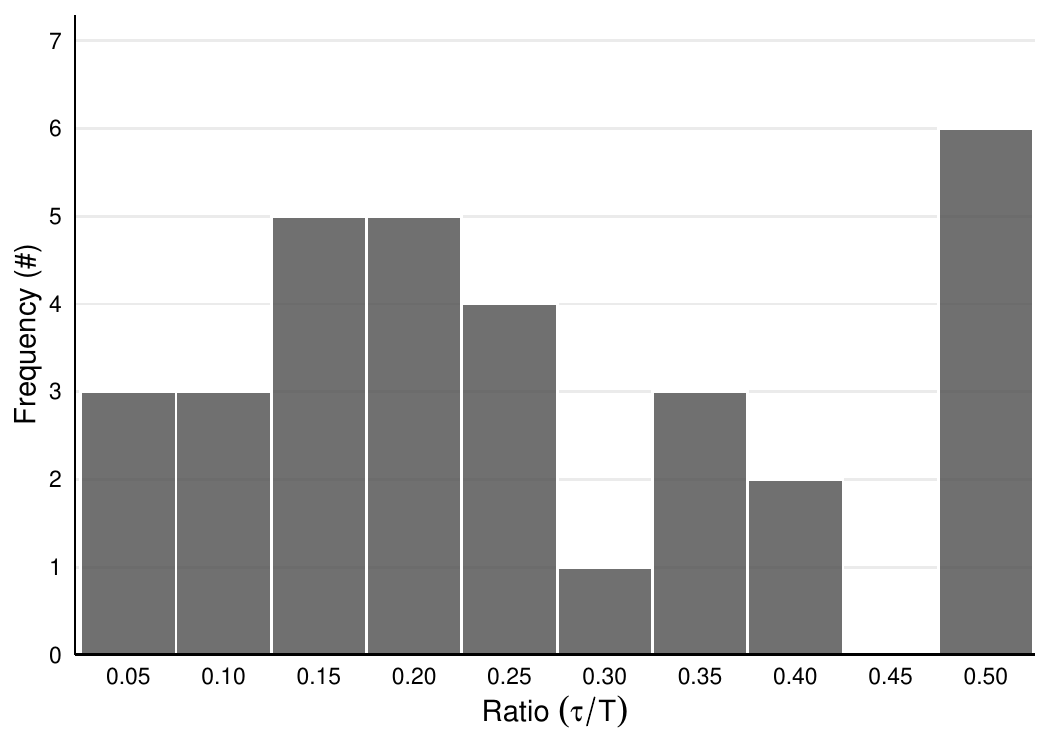}\\
   (a) Choices of $\tau$ by panel length ($T$) & (b) Histogram of the ratio of $\tau/T$ \\
    \end{tabular}
    \caption{Summary of choices of long-differencing averaging window $\tau$ across the literature}
    \label{fig:ld_lit}
        \raggedright
    \begin{singlespace}
    {\footnotesize \textit{Notes:} This figure plots 32 specifications choices for $\tau$ with corresponding panel lengths ($T$) from 29 research articles using LD between 2012 and 2026.} 
    \end{singlespace}
\end{figure}

 Differences between ${\beta}_{FE}$ and ${\beta}_{LD}$ are attributed to adaptation. The economic reasoning behind this comparison is that the fixed effect estimand $\beta_{FE}$ measures the marginal responses to weather where certain long-run adaptations, such as technology or capital investments, are not possible \citep{lemoine2018nber,burke2016aejep}.\footnote{See \citet{lemoine2018nber} for detailed analysis of the margins of adaptation captured by panel first- and long-difference estimators in a dynamic environment where agents maximize expected payoffs that depend on a stock variable.} By contrast, the LD estimand reflects a longer horizon where it is assumed that agents are able to adjust along additional margins. Under this framework, the null hypothesis of no adaptation is given by: 
\begin{equation}
    H_0: \beta_{FE}=\beta_{LD}.
\end{equation}

A rejection of $H_0$ implies that the long-run response to climate is statistically different from the short-run response. The measured adaptive behavior depends on the direction of the difference. In the often-studied cases where larger marginal effects of treatment imply adverse outcomes (i.e.\ heat on mortality or crop yields), a long-run response that is smaller than short-run ($|\beta_{LD}|<|\beta_{FE}|$) indicates long-run adaptation.\footnote{It may be worth noting that this test does not speak to the \textit{cost} of adaptation.} Conversely, a long-run response that is larger ($|\beta_{LD}|>|\beta_{FE}|$) indicates intensification of climate impacts exceeds agents' ability to adapt.

\subsection{What are LD and FE consistent for when long-run and short-run response differ?}
\label{sec:theory_results}

In this section, we show the probability limits of the FE and LD estimators. While the difference between the probability limit of a given estimator and its target parameter is technically an inconsistency, we will often refer to it as bias for simplicity. We provide the proofs for the following propositions in Appendix \ref{appendix:proofs}. In addition to standard moment conditions required for the probability limits to be well-defined, we impose strict exogeneity throughout the paper to focus our analysis on the bias stemming from $\theta_L\neq \theta_S$.

\begin{proposition} \label{prop:FE} 
Suppose that $\bar{E}[\sum_t\tilde{x}_{it}\tilde{y}_{it}]<\infty$ and $0<\bar{E}[\sum_t\tilde{x}_{it}^2]<\infty$. Suppose further that Eq.\ \eqref{eq:outcome_model} holds and $E[\varepsilon_{it}|X_i] = 0$ for all $i=1,\dots,n$ and $t=1,\dots,T$. 
\begin{enumerate}[(i)]
\item $\beta_{FE}=\omega_{S}^{FE}\theta_S+\omega_{L}^{FE}\theta_L$, where 
\begin{eqnarray}\omega_S^{FE}\equiv\frac{\bar{E}[\sum_t\tilde{w}^{*2}_{it}]+\bar{E}[\sum_t\tilde{w}^{*}_{it}\tilde{c}^{*}_{it}]}{\bar{E}[\sum_t\tilde{x}_{it}^2]}\quad \omega_L^{FE}\equiv \frac{\bar{E}[\sum_t\tilde{c}^{*2}_{it}]+\bar{E}[\sum_t\tilde{w}^{*}_{it}\tilde{c}^{*}_{it}]}{\bar{E}[\sum_t\tilde{x}_{it}^2]}\end{eqnarray}
and $\omega_{S}^{FE}+\omega_{L}^{FE}=1$.
\item $\beta_{FE}-\theta_S=(\theta_L-\theta_S)\omega_L^{FE}$.
\end{enumerate}
\end{proposition}

Part (i) of Proposition \ref{prop:FE} shows that the FE estimand is a weighted average of short- and long-run responses to weather and climate, $\theta_S$ and $\theta_L$, respectively. The weight on $\theta_S$, $\omega_S^{FE}$, depends on the within-demeaned weather shocks $\tilde w^*_{it}$ and its covariance with the within-demeaned climate.  Intuitively, if climate variation is eliminated by the within-group transformation, this weight equals 1. Part (ii) of Proposition \ref{prop:FE} demonstrates that the weight on the long-run response $\omega_L^{FE}$ biases the FE estimand away from the short-run response and depends on the variation from the unobserved component of climate that survives the within-transformed observed weather ($\tilde x_{it}$). We will therefore refer to $\omega_L^{FE}$ as the FE contamination weight. The resulting bias therefore depends on both the magnitude of $\omega_L^{FE}$ and the difference between $\theta_L$ and $\theta_S$. This resonates with recent work challenging the notion that FE estimators solely capture short-run response since agents can adapt to expectations of change \citep{lemoine2018nber,shrader2023wp}.

\begin{proposition} \label{prop:LD} Suppose that $\bar{E}[\Delta\bar{x}_{i}\Delta\bar{y}_{i}]<\infty$ and $0<\bar{E}[\Delta\bar{x}_{i}^2]<\infty$. Suppose further that Eq.\ \eqref{eq:outcome_model} holds and $E[\varepsilon_{it}|X_i] = 0$ for all $i=1,\dots,n$ and $t=1,\dots,T$. 
    \begin{enumerate}[(i)]
\item $\beta_{LD}=\omega_S^{LD}\theta_S+\omega_L^{LD}\theta_L$, where
\begin{eqnarray}\omega_S^{LD}\equiv\frac{\bar{E}[\Delta \bar{w}_{i}^{*2}]+\bar{E}[\Delta \bar{w}_{i}^{*} \Delta \bar{c}_{i}^{*}]}{\bar{E}[\Delta \bar{x}_{i}^2]},\quad \omega_L^{LD}\equiv \frac{\bar{E}[\Delta \bar{c}_{i}^{*2}]+\bar{E}[\Delta \bar{w}_{i}^{*} \Delta \bar{c}_{i}^{*}]}{\bar{E}[\Delta \bar{x}_{i}^2]},\end{eqnarray}
and $\omega_{S}^{LD}+\omega_{L}^{LD}=1$.
\item $\beta_{LD}-\theta_L=(\theta_S-\theta_L)\omega_S^{LD}$
\end{enumerate}
\end{proposition}

Part (i) of Proposition \ref{prop:LD} shows that the LD estimand is also a weighted average of the short- and long-run parameters. The weight from short-run variation, $\omega_S^{LD}$, depends on the extent of the variation from the unobserved weather shocks $w^*_{it}$ that survives the LD transformation. We will refer to $\omega_S^{LD}$ as the LD contamination weight.  Part (ii) of Proposition \ref{prop:LD} demonstrates that the bias of $\beta_{LD}$ depends on both the magnitude of $\omega_S^{LD}$ and the difference between short- and long-run response parameters. 

To provide intuition on the weight in $\omega_S^{LD}$, the bias of $\beta_{LD}$ in Proposition \ref{prop:LD}, it is helpful to consider the case where the long-differenced climate and weather shocks are uncorrelated, $\bar{E}[\Delta\bar w_i^*\Delta \bar c_i^*]=0$. In this case, $\omega_S^{LD}$ depends inversely on the climate signal-to-noise ratio. To see this, first we simplify $\omega_S^{LD}$ under the zero-covariance assumption and then divide both numerator and denominator by $\bar{E}[\Delta \bar{w}_i^{*2}]$, which yields the following, 
\begin{eqnarray}
    \omega_S^{LD}&=&\frac{\bar{E}[\Delta \bar{w}_i^{*2}]}{\bar{E}[\Delta \bar{c}_i^{*2}]+\bar{E}[\Delta \bar{w}_i^{*2}]}=\frac{1}{1+SNR_c}
\end{eqnarray}
where we use $SNR_c\equiv \bar{E}[\Delta \bar{c}_i^{*2}]/\bar{E}[\Delta \bar{w}_i^{*2}]$ to denote the climate signal-to-noise ratio assuming $\bar{E}[\Delta \bar{w}_i^{*2}]> 0$.
Clearly, $\omega_S^{LD}$ (and consequently the bias of $\beta_{LD}$) decreases as $SNR_c$ increases. In Section \ref{sec:choice_tau}, we examine how the choice of $\tau$, the LD averaging window, affects the magnitude of $\omega_S^{LD}$.

The formulae in Propositions \ref{prop:FE} and \ref{prop:LD} highlight the potential for bias in both estimators of short- and long-run responses. Importantly, the weights and resulting bias depend on components of the data-generating process unobservable to the researcher, which we analyze numerically in Section \ref{sec:simulations}.

\subsection{Attenuation bias in LD and FE estimation}

Given the natural connection to measurement error problems, we next examine the conditions under which our setting yields attenuation bias similar to classical measurement error problems. To obtain this result, we impose plausible non-negative covariance assumptions. Specifically, we assume that LD- and FE-transformed weather shocks and climate are non-negatively correlated. The climate literature often considers transitory inter-annual weather variability and long-run climate trends to be physically independent phenomena driven by distinct mechanisms \citep{marotzke2015nature,maher2018grl,lehner2023erc,hasselmann1976tf}, which imply that the two variables are uncorrelated.  While this covariance might temporarily be ambiguous in short panels due to cyclical patterns like El Niño, over longer horizons this correlation dissipates (or is likely positive). Let $Sgn(x)$ denote the function that yields the sign of its argument $x$, specifically $Sgn(x)=1\{x>0\}-1\{x<0\}$.
\begin{corollary}[Attenuation bias]  \label{cor:attenuation}Suppose the assumptions in Propositions \ref{prop:FE} and \ref{prop:LD} hold. Furthermore, suppose that $\bar{E}[\tilde{w}_{it}^*\tilde{c}_{it}^*]\geq 0$ and $\bar{E}[\Delta \bar{w}_i^*\Delta \bar{c}_i^*]\geq 0$,

\begin{enumerate}[(i)] \item $|\beta_{FE}-\theta_S|\leq |\theta_L-\theta_S|$ and $Sgn(\beta_{FE}-\theta_S)=Sgn(\theta_L-\theta_S)$
 
\item $|\beta_{LD}-\theta_L|\leq |\theta_S-\theta_L|$ and $Sgn(\beta_{LD}-\theta_L)=Sgn(\theta_S-\theta_L)$.

\item $Sgn(\beta_{FE}-\theta_S)=-Sgn(\beta_{LD}-\theta_L)$

\end{enumerate}
\end{corollary}

The sign restriction on the covariances in Corollary \ref{cor:attenuation} allows us to place bounds on the magnitude and direction of the FE and LD biases relative to their respective targets. The magnitude of the bias of each estimator is weakly smaller than the difference between the short- and long-run population estimands. The third result in Corollary \ref{cor:attenuation} states that the biases are of opposite signs. The implication is that the biases demonstrated in Propositions \ref{prop:FE} and \ref{prop:LD} attenuate differences between the FE and LD estimands away from their respective targets.
\subsection{Testing and quantifying adaptation}
The following corollary building on Propositions \ref{prop:FE} and \ref{prop:LD} characterizes the relationship between $\beta_{LD}-\beta_{FE}$ and the extent of adaptation, $\theta_L-\theta_S$. This characterization has two practical takeaways. First, testing adaptation using the null hypothesis, $H_0: \beta_{FE}=\beta_{LD}$, constitutes a test by implication. Second, to estimate $\theta_L-\theta_S$ consistently, one has to estimate the contamination weights, which require an assumption on the climate process. 
\begin{corollary}\label{cor:adaptation_test}
Suppose the assumptions in Propositions \ref{prop:FE} and \ref{prop:LD} hold.\\ $\beta_{LD}-\beta_{FE}=(\theta_L - \theta_S)(1-\omega_S^{LD} - \omega^{FE}_L)$
 \end{corollary}

Corollary \ref{cor:adaptation_test} characterizes the difference between the two estimands, which form the basis of the adaptation test. The difference in estimands is linear in $(\theta_L-\theta_S)$ with a slope equal to $(1 - \omega_S^{LD} - \omega_L^{FE})$. Since $\omega_S^{LD}$ and $\omega_L^{FE}$ depend on the data-generating process of $c_{it}^*$ and $w_{it}^*$, this corollary highlights the consequences of LD and FE comparisons being agnostic about the climate process. It also demonstrates that if researchers are willing to impose assumptions on this process, then they can conduct inference on $\theta_L-\theta_S$ using Corollary \ref{cor:adaptation_test}. We discuss this implication further in Section \ref{sec:implications}.

Corollary \ref{cor:adaptation_test} further demonstrates that the test of $H_0: \beta_{LD}=\beta_{FE}$ is based on a necessary, but not sufficient, condition of $H_0:\theta_L=\theta_S$. It therefore controls size, but may have trivial power under the alternative if $\omega_S^{LD}+\omega_L^{FE}=1$. Furthermore, even if $\omega_S^{LD}+\omega_L^{FE}\neq 1$ but close to it, then this can compromise the power of the test as we demonstrate numerically in the following section. Intuitively, the equality $\beta_{FE} = \beta_{LD}$ is implied by the
no-adaptation hypothesis $\theta_L = \theta_S$, but it is not equivalent to it:
by Corollary \ref{cor:adaptation_test} the same equality can arise under adaptation
whenever $\omega^{LD}_S + \omega^{FE}_L = 1$. A non-rejection is therefore
consistent with either no adaptation or with adaptation masked by contamination.
This is the same logic as other tests by implication, such as
over-identification tests.

\begin{remark}[FE and LD under response heterogeneity]\label{rem:heterogeneity}
In Appendix \ref{app:heterogeneity}, we provide probability limits for $\beta_{FE}$ and $\beta_{LD}$ when allowing for cross-sectional response heterogeneity through the following correlated random coefficient model, specifically
$$y_{it}=\theta_{S,i}w_{it}^*+\theta_{L,i}c_{it}^*+\alpha_i+\varepsilon_{it}.$$
This mild deviation from homogeneous response complicates the probability limits. For the LD estimator, the probability limit consists of two components:
\begin{eqnarray*}\beta_{LD}&=&\underbrace{\bar{E}\left[\theta_{L,i}\frac{E[\Delta \bar x_i^2|\theta_{L,i},\theta_{S,i}]}{\bar{E}[\Delta \bar{x}_{i}^2]}\right]}_{\text{variance-weighted average of $\theta_{L,i}$}}+\underbrace{\bar{E}\left[\frac{E[\Delta \bar x_i^2|\theta_{L,i},\theta_{S,i}]}{\bar{E}[\Delta \bar{x}_{i}^2]}\omega_{S,i}^{LD}(\theta_{S,i}-\theta_{L,i})\right]}_{\text{contamination from $\Delta \bar w_i^*$}}
\end{eqnarray*}
where $\omega_{S,i}^{LD}\equiv\frac{E[\Delta \bar{w}_i^{*2}|\theta_{L,i},\theta_{S,i}]+E[\Delta\bar{c}_i\Delta\bar{w}_i|\theta_{L,i},\theta_{S,i}]}{E[\Delta \bar{x}_{i}^2|\theta_{L,i},\theta_{S,i}]}$ assuming $E[\Delta \bar{x}_{i}^2|\theta_{L,i},\theta_{S,i}]>0$. 

The first component of $\beta_{LD}$ is a variance-weighted average similar to other contexts where response heterogeneity is ignored in a fixed effects estimand \citep[e.g.][]{gibbons2019,ghanem2021what}. As a result, even in the absence of contamination from $\Delta\bar w_i^*$, the LD estimand captures a weighted average that gives higher weight to cross-sectional units that have higher variability in long-differenced weather, $\Delta\bar x_i$, relative to those with less variability. 

The second component captures the contamination from $\Delta \bar w_i^*$. This term depends on the unit-specific variance weight, LD contamination weight ($\omega_{S,i}^{LD}$) and $\theta_{S,i}-\theta_{L,i}$. To simplify this term, suppose that $\omega_{S,i}^{LD}$ is a strictly positive constant for all $i$, then the second term is the variance-weighted average of $\theta_{S,i}-\theta_{L,i}$. Suppose that cross-sectional units that experienced higher weather variability exhibit higher adaptation and thereby have a larger difference between $\theta_{S,i}$ and $\theta_{L,i}$, then this would exacerbate the magnitude of the second component, assuming the sign of $\theta_{S,i}-\theta_{L,i}$ is the same for all $i$.
    \end{remark}

\subsection{Empirically-calibrated simulation study}\label{sec:simulations}

We illustrate the formal analysis in a simulation design calibrated to the empirical setting in \citet{burke2016aejep}. The authors study the sensitivity of crop yields to temperature, finding negative effects of degree days over 29$^\circ$ Celsius with FE and LD approaches. Using the observed temperature data, we conduct simulations with $c_{it}^*$ defined as a climate normal, $c_{it}^*=\frac{1}{m}\sum_{\ell=1}^mx_{i(t-\ell)}$, to evaluate the performance of various specifications of the studied estimators.
\subsubsection{Simulation Design}  
We reconstruct the temperature dataset used in \citet{burke2016aejep} to measure adaptation of corn yields to extreme heat. Table \ref{tab:simulation_design_1} presents the details of the data-generating process. We construct a dataset of county-level daily temperature for a longer sample period than the one used in \citet{burke2016aejep} (1950-2022) (See Appendix \ref{app_section:sim1_data}). The additional data allow us to rely on the climate normal specification of ($c_{it}^*$) and compare our formal results across various specification choices for the LD estimator.
\begin{table}[H]
\centering
\footnotesize
\begin{threeparttable}
\caption{Simulation Design: Data-Generating Process}
\label{tab:simulation_design_1}
\begin{tabular}{@{}p{4cm}p{8cm}@{}}
\midrule
Outcome: & $Y_{it} = \theta_S w^*_{it} + \theta_L c^*_{it} + \eta_i + \delta_t + \varepsilon_{it} $\\
&$\text{for } i = 1,\dots,N,\ t = 1,\dots,T$ \\
Observed weather $x_{it}$: & Degree days above 29$^\circ$C \citep{burke2016aejep}\\
Climate $c_{it}^*$:&$c_{it}^*=\frac{1}{m}\sum_{\ell=1}^mx_{i(t-\ell)}$\\
Parameters: & $\theta_S = -0.07,\ \theta_L \in \{-0.07, -0.06,..., 0\}$ \\
Units: & N = 1,000\\

\midrule
\end{tabular}
\end{threeparttable}
\label{sim_results}
\end{table}

\begin{figure}[htbp]
    \centering
\begin{tabular}{c}\includegraphics[width=0.65\linewidth]{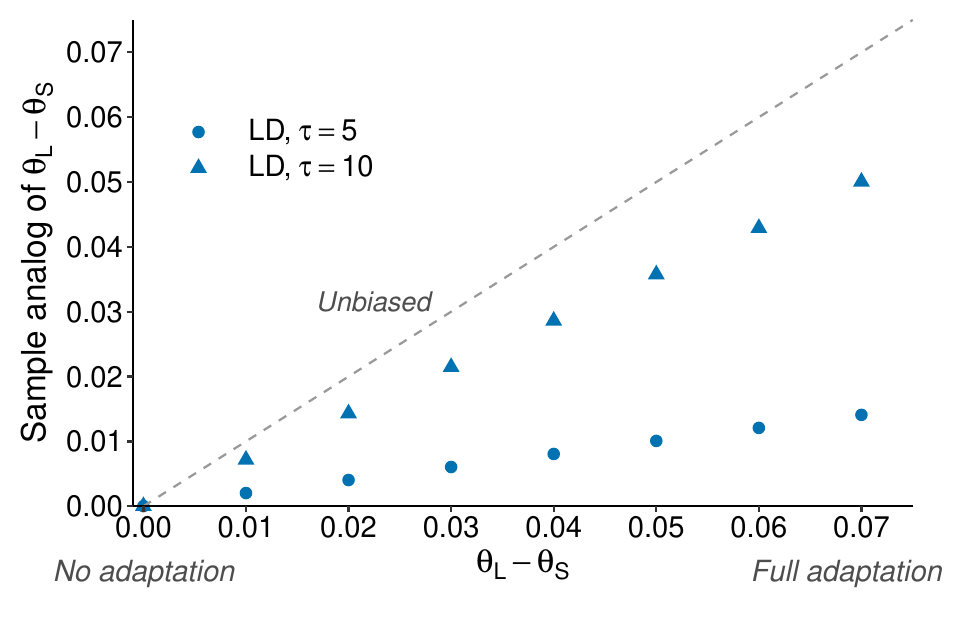}
    
    \\
    (a) Simulation Mean of $\hat{\beta}_{LD}-\hat{\beta}_{FE}$\\
    \\
    \\
        \includegraphics[width=0.65\linewidth]{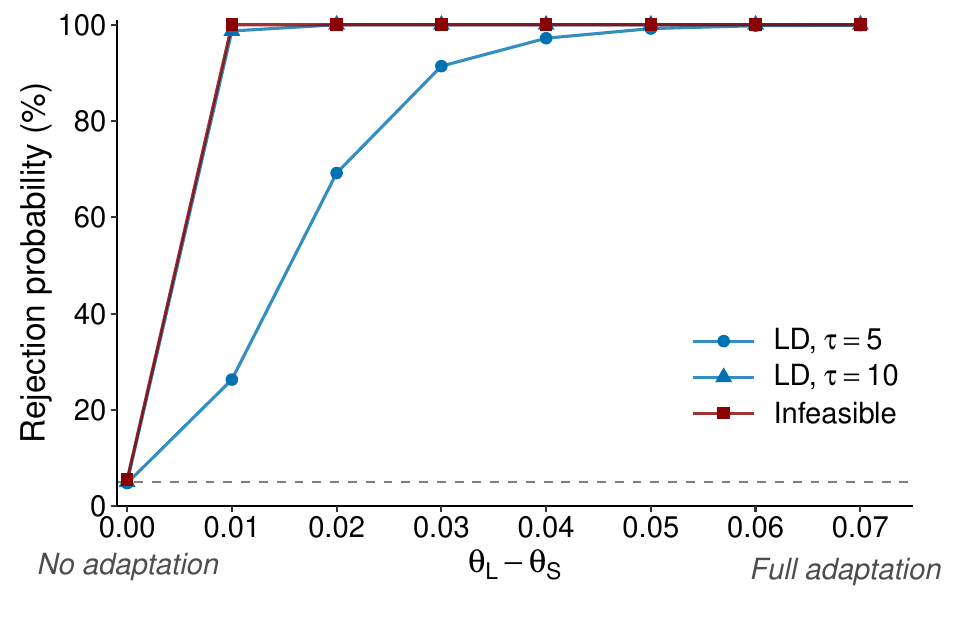}\\
    
    (b) Simulation Rejection Probability of Adaptation Test\\
    \end{tabular}
    
    \caption{Simulation Results: FE, LD and Adaptation Test}
    \label{fig:main simulation_results}
    \caption*{\scriptsize \textit{Notes:} Panel (a) plots the difference of the simulation estimate means $\bar{\hat \beta}_{LD}-\bar{\hat\beta}_{FE}$ against $\theta_{L}-\theta_{S}$ for long-differences using an averaging window ($\tau$) of 5 and 10. Panel (b) plots the rejection probability curves for LD ($H_0: \beta_{LD} = \beta_{FE}$). The null hypothesis of the infeasible test is $H_0:\theta_L=\theta_S$, which are estimated using $c_{it}^*$ and $w_{it}^*$ as regressors in a fixed-effects model. Simulation rejection probabilities are computed across 1,000 replications for  for a grid of values of $\theta_L$ over $\{-0.07,-0.06,0\}$ with $\theta_S=-0.07$.}
\end{figure}

\subsubsection{FE, LD and Adaptation Test}

Figure \ref{fig:main simulation_results} presents the simulation mean of $\hat{\beta}_{LD}-\hat{\beta}_{FE}$ and  the simulation rejection probabilities of the adaptation test ($H_0: \beta_{FE}=\beta_{LD}$). We vary $\theta_L$ over the grid $\{-0.07,-0.06, \dots, 0\}$, where we fix $\theta_S=-0.07$. As a result, the range of values we consider for $\theta_L$ starts from the no-adaptation case ($\theta_L=\theta_S$) to the full adaptation case ($\theta_L=0$), where the degree days above 29$^\circ$C are no longer harmful.\footnote{The FE estimate in \citet{burke2016aejep} is $-0.07$, whereas the LD estimate is $-0.02$ (see Table \ref{table:validate_data_be} for a replication using our constructed sample.)} The results presented in Figure \ref{fig:main simulation_results} correspond to the specification of $c_{it}^*=\frac{1}{m}\sum_{\ell=1}^mx_{i(t-\ell)}$ with $m=10$ (10-year climate normal), $T=25$ and $\tau\in\{5,10\}$. Appendix \ref{app:sim_results} provides a broader set of simulation statistics for this variant of the simulation design as well as for other variants with different choices of $m$ and $T$.

Panel (a) of Figure \ref{fig:main simulation_results} plots the simulation mean $\hat{\beta}_{LD}-\hat{\beta}_{FE}$ as a function of $\theta_L-\theta_S$, the extent of adaptation in the true DGP. Consider the case with $\tau=5$, the simulation mean of $\hat{\beta}_{LD}-\hat{\beta}_{FE}$ is about 20\% of $\theta_L-\theta_S$. As per Corollary \ref{cor:adaptation_test}, this proportion should be explained by $(1-\omega_S^{LD}-\omega_L^{FE})$. Indeed, $\hat{\beta}_{LD}-\hat{\beta}_{FE}$ is linear in $\theta_L-\theta_S$ with a slope of about $0.2$, since simulation means of $\hat{\omega}_S^{LD}$ and $\hat{\omega}_L^{FE}$ equal to 0.825 and -0.0265, respectively (see Table \ref{tab:simulation1_variantA}).\footnote{Note that we can use these results from Table \ref{tab:simulation1_variantA}, even though it reports the simulation results for $\theta_L=-0.02$ and $\theta_S=-0.07$, since $\omega_S^{LD}$ and $\omega_L^{FE}$ do not depend on the values of $\theta_L$ and $\theta_S$.} When $\tau$ is increased, however, the slope of the line increases to about 0.7, since the simulation mean of $\hat{\omega}_S^{LD}$ and $\hat{\omega}_L^{FE}$ equal 0.3119 and -0.0055, respectively (see Table \ref{tab:simulation1_variantA}). This simulation study demonstrates that the contamination weight $\omega_S^{LD}$ tends to be larger in magnitude relative to $\omega_L^{FE}$.\footnote{This demonstrates that the assumption made in \citet{carter2018arre} which suggests that the LD estimator is consistent for $\theta_L$ is implausible for temperature data and realistic climate models.} The asymmetry between the two contamination weights has a simple source. The within transformation differences out unit-specific means and leaves
predominantly high-frequency variation, so little climate variation survives and
$\omega^{FE}_L$ is small. The LD transformation instead averages
within windows to cancel weather, but the variance of the surviving shock falls
with $\tau$---$\bar{E}[\Delta \bar{w}^{*2}_i] =
2\sigma^2_w/\tau$, for example, if $w_{it}^*$ are i.i.d.\ shocks---thus in finite windows a non-trivial share of weather remains
and $\omega^{LD}_S$ is comparatively large. Averaging is thus a blunter
instrument for isolating climate than differencing is for isolating weather shocks.

In addition, our results here demonstrate a case of attenuation bias per Corollary \ref{cor:attenuation}, where both $\hat{\beta}_{LD}$ and $\hat{\beta}_{FE}$ under-estimate $\theta_L$ and $\theta_S$, respectively (see Table \ref{tab:simulation1_variantA}). We note however, that in other variants of our simulation design we find $\tau$ values that lead to over-estimation of $\theta_L-\theta_S$ (see Panel (b) of Figure \ref{fig:c10_ld_summary}).

Panel (b) of Figure \ref{fig:main simulation_results} presents the simulation rejection probabilities of the adaptation test ($H_0: \beta_{FE}=\beta_{LD}$). This figure demonstrates that the larger bias when using $\tau=5$ relative to $\tau=10$ can have substantive power consequences. We caution, however, that the choice of $\tau$ may also have an impact on the sampling variability of the LD estimator, which can have power implications.\footnote{For instance, for the case with $c_{it}^*$ as the 30-year normal, using $\tau=5$ leads to a higher rejection probability relative to $\tau=10$ (see Figure \ref{fig:sim_rej_cov_30-year}), as the former is associated with a higher simulation standard deviation than the latter.}

\section{Climate signal-to-noise ratio and the choice of $\tau$ in long-differencing}\label{sec:choice_tau}

In this section, we demonstrate that the extent of contamination from short-run weather variation in the long-differenced climate is inversely related to the climate signal-to-noise ratio and analyze this object and its dependence on the choice of $\tau$ in the long-differencing approach. We do so with the aid of two analytical examples of climate specifications of $x_{it}$, where we can solve for the climate signal-to-noise ratio analytically as a function of $T$, $\tau$, and other parameters of the data-generating process. We also compute the climate signal-to-noise ratio using temperature data.

First, suppose that $c_{it}^*=\mu_i+\delta_it$ and $w_{it}^*|\mu_i,\delta_i\overset{i.i.d.}{\sim} (0,\sigma_w^2)$, then in this case
the climate signal-to-noise ratio simplifies to
\begin{eqnarray}SNR_c&=& \left[\tau(T-\tau)^2 \frac{\sigma_{\delta}^2}{2 \sigma_{w}^2} \right]\label{eq:SNR_trend}\end{eqnarray}
where $\sigma_{\delta}^2\equiv Var(\delta_i)$. For a formal derivation of this result with all relevant assumptions, see Proposition \ref{prop:tau_for_trend} and its proof.

The main takeaways from Eq.\ \eqref{eq:SNR_trend} are that $SNR_c$ is increasing in the time horizon $T$ and $\sigma_\delta^2/\sigma_w^2$. As a result, for $\omega_S^{LD}\rightarrow 0$, either $T\rightarrow \infty$ and/or $\sigma_\delta^2/\sigma_w^2\rightarrow \infty$. Panel (a) in Figure \ref{fig:SNR} demonstrates however that for fixed $T$ and $\sigma_\delta^2/\sigma_w^2$, $SNR_c$ depends nonlinearly on $\tau$.\footnote{This nonlinearity stems from a trade-off in how $\tau$ affects the variance of the long-differenced climate and weather shocks, as evidenced by their respective variance formulae derived in Proposition \ref{prop:tau_for_trend}
$\bar{E}[\Delta \bar{c}_i^{*2}]=\sigma_\delta^2(T-\tau)^2$ and 
    $\bar{E}[\Delta \bar{w}_i^{*2}]=\frac{2\sigma_w^2}{\tau}$.} For this specification of $c_{it}^*$, there exists $\tau\in(1,T/2)$ that maximizes $SNR_c$, specifically $\tau^*=T/3$. It is important to note, however, that even if we use the $SNR_c$-maximizing $\tau^*$, this does not mean that $\omega_S^{LD}=0$ and thereby $\beta_{LD}=\theta_L$. Indeed, Figure \ref{fig:SNR} demonstrates that the magnitude of $SNR_c$ and subsequently $\omega_S^{LD}$ even at the optimal $\tau^*$ will ultimately depend on the value of $T$ and $\sigma_\delta^2/\sigma_w^2$.

\begin{figure}[htbp]
    \centering
    \begin{tabular}{cc}
    \includegraphics[width=0.45\linewidth]{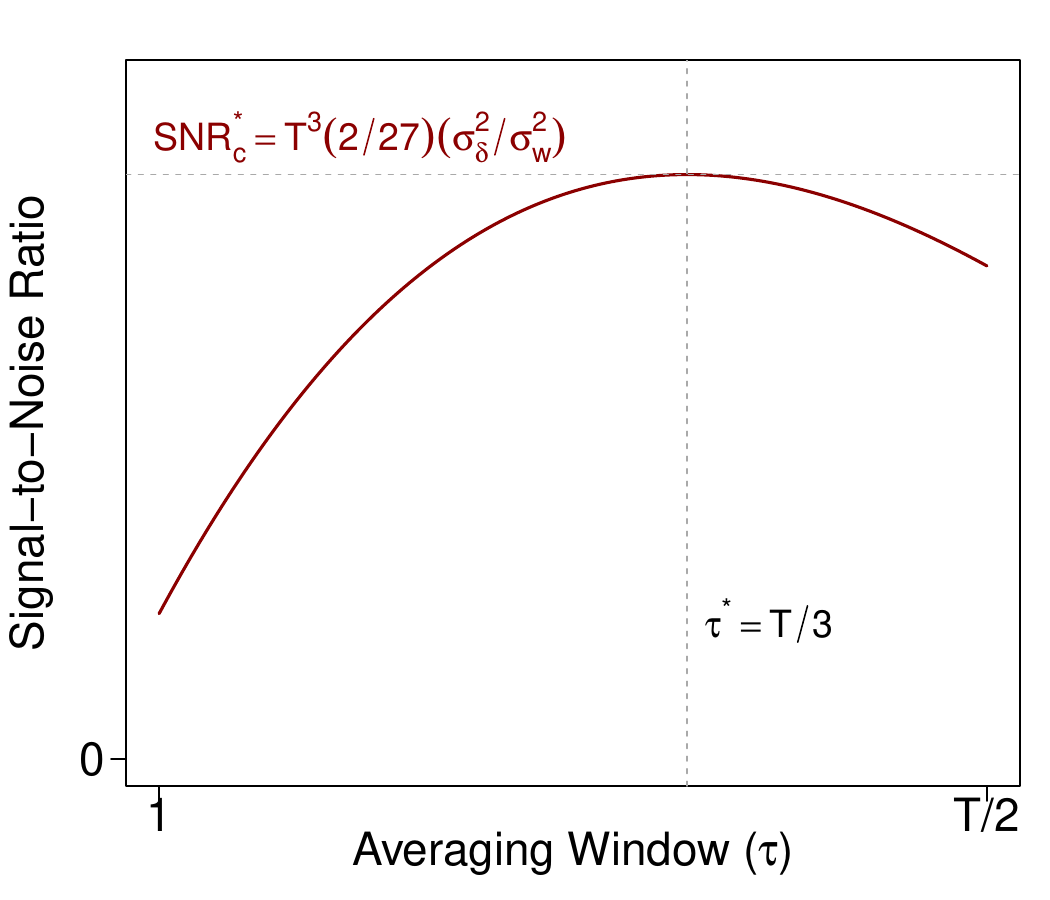}&  \includegraphics[width=0.45\linewidth]{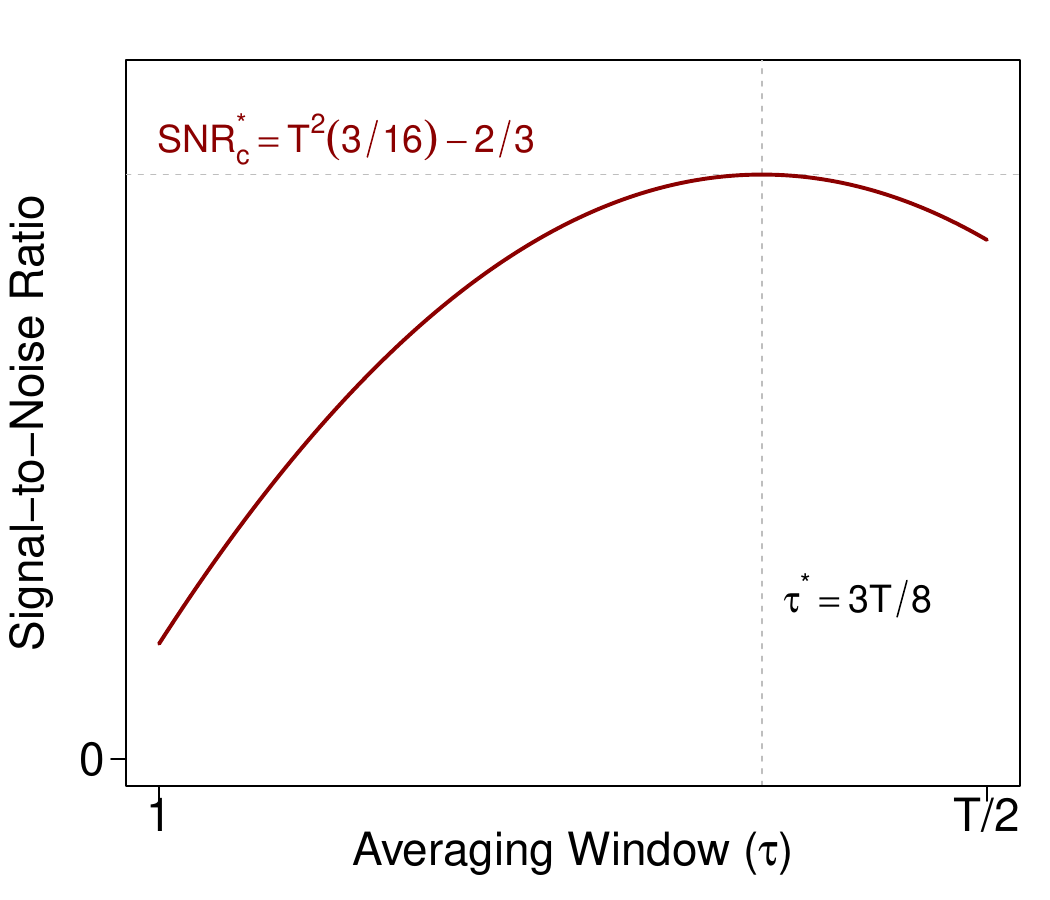}\\
   (a) Linearly Trending Case& (b) Unit-Root Case\\
    \end{tabular}
    \caption{Climate signal-to-noise ratio as a function of $\tau$ ($T=25$)}
    \label{fig:SNR}
\end{figure}

Since climate normals, popular in climate modeling, give rise to a unit-root AR(m), where $x_{it}=\frac{1}{m}\sum_{\ell=1}^mx_{i(t-\ell)}+\eta_{it}$ and $c_{it}^*=\frac{1}{m}\sum_{\ell=1}^mx_{i(t-\ell)}$, we also consider the unit root example. To simplify illustration, we first consider the special case where $m=1$ and thereby $c_{it}^*=x_{i(t-1)}$\footnote{This is an empirically relevant case if economic agents use lagged weather as their expectation of weather (climate).} and consider the more general climate normal case using the temperature data we use in our simulations in Figure \ref{fig:sim_SNR}. 

Similar to the linearly trending case, the unit-root example also demonstrates a trade-off in the choice of $\tau$. In this case, a larger $\tau$ decreases both the variance of $\Delta\bar w_i^*$ and $\Delta \bar c_i^*$. The climate signal-to-noise ratio therefore depends nonlinearly on $\tau$ as demonstrated clearly in Panel (b) of Figure \ref{fig:SNR}. Since $E[\Delta \bar w_i^*\Delta\bar c_i^*]\neq 0$ in this case, the signal-to-noise ratio is given by
\begin{eqnarray}
    SNR_c&=&\frac{\bar{E}[\Delta\bar c_i^{*2}]+\bar{E}[\Delta \bar c_i^*\Delta\bar w_i^*]}{\bar{E}[\Delta\bar w_i^{*2}]+\bar{E}[\Delta \bar c_i^*\Delta\bar w_i^*]}=T\tau - \frac{4\tau^2}{3} + \frac{\tau}{3}-1.
\end{eqnarray}
Proposition \ref{prop:tau_unit_root} and its proof state the assumptions and provided a detailed derivation of the signal-to-noise ratio in the unit root case as well as the choice of $\tau$ that maximizes this ratio and thereby minimizes $\omega_S^{LD}$.

    \begin{figure}[htbp]
    \centering   
    \begin{subfigure}[t]{0.45\linewidth}
\includegraphics[width=\linewidth]{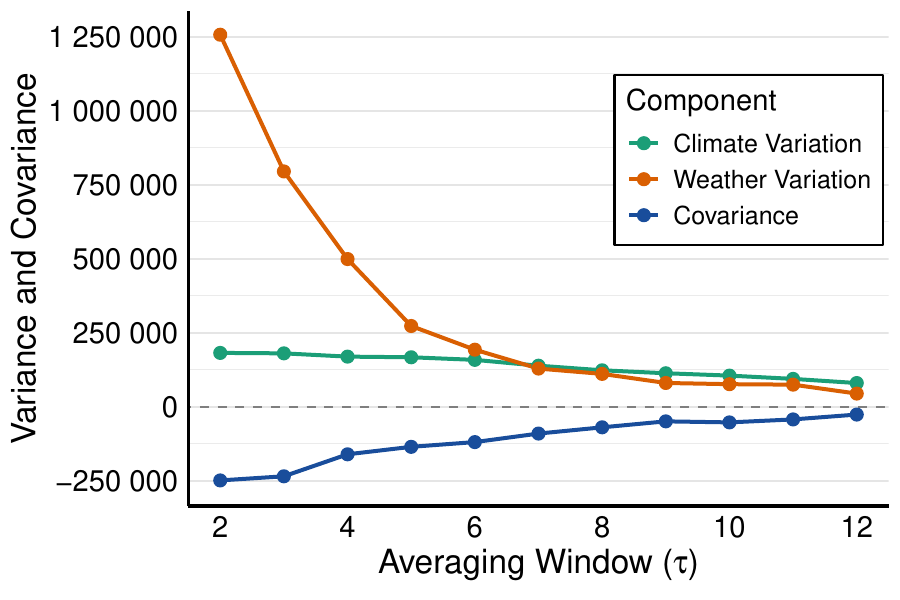}    
\caption{$SNR_c$ components with 10-year normal}\end{subfigure}    
\begin{subfigure}[t]{0.45\linewidth}
\includegraphics[width=\linewidth]{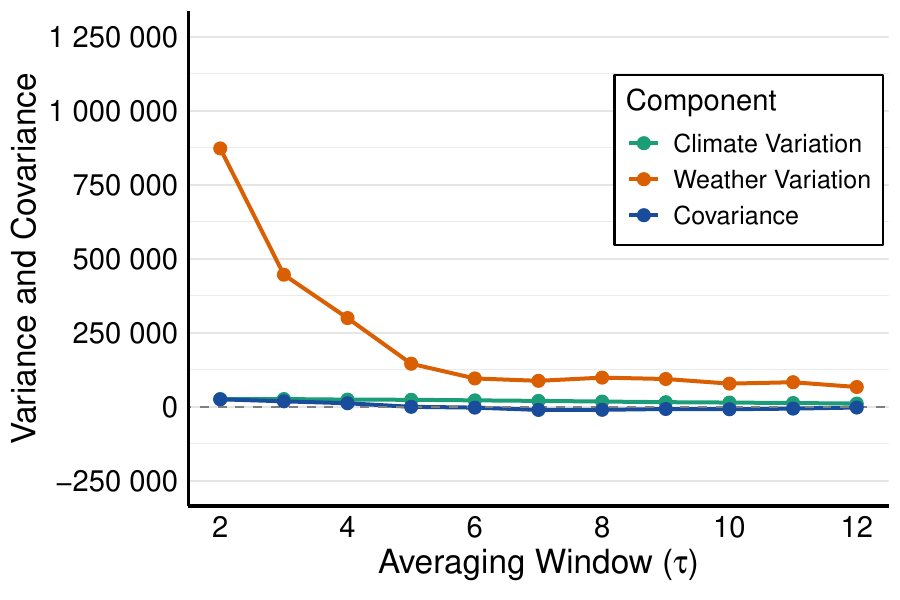}    
\caption{$SNR_c$ components with 30-year normal}\end{subfigure}
 \begin{subfigure}[t]{0.45\linewidth}
            \includegraphics[width=\linewidth]{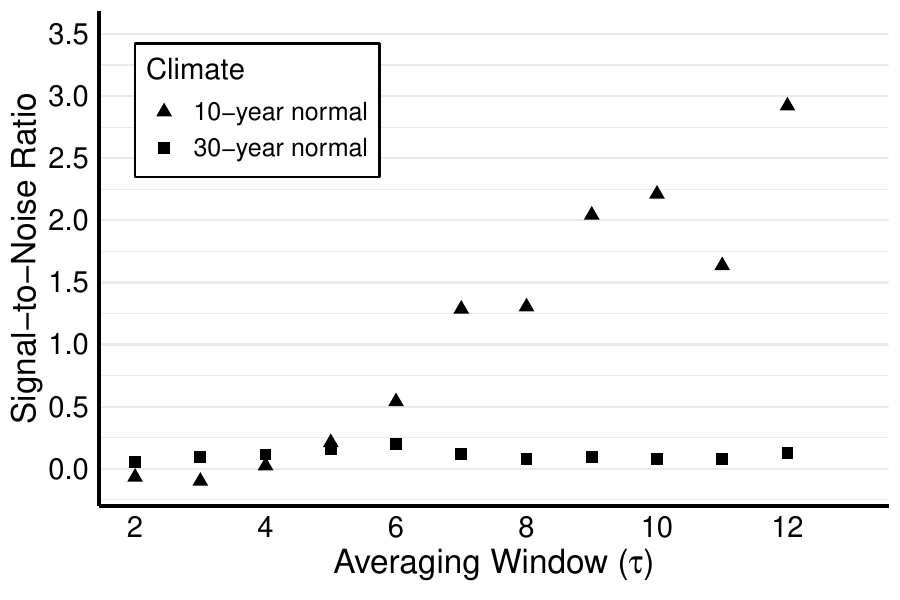}
        \caption{Signal-to-Noise Ratio}
       \end{subfigure}%
          \begin{subfigure}[t]{0.45\linewidth}
                       \includegraphics[width=\linewidth]{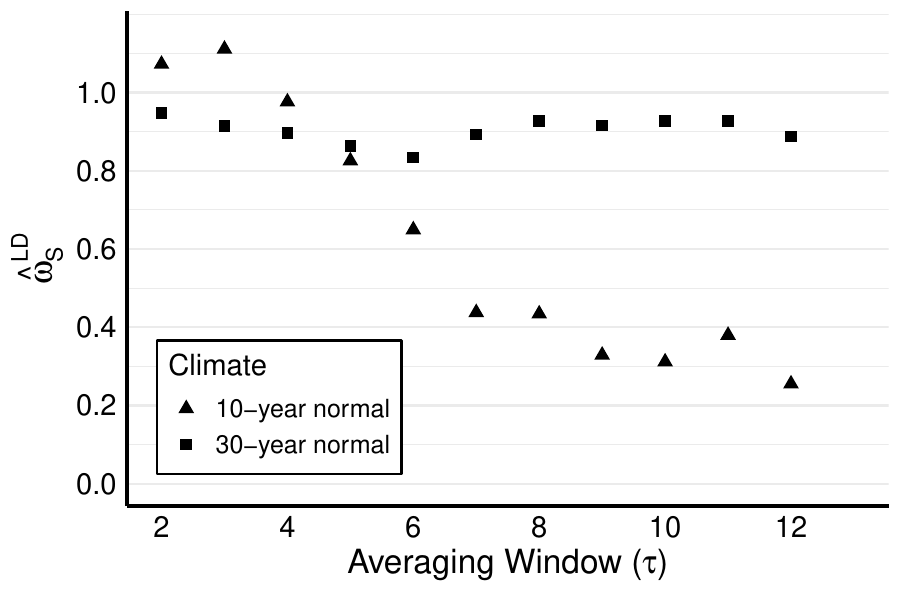}
        \caption{LD contamination weight ($\omega_S^{LD}$)}
      
       \end{subfigure}

    \caption{Sample estimates of $SNR_c$ across $\tau$ for different climate normals with T = 25} \label{fig:sim_SNR} 
    \caption*{\scriptsize{\emph{Notes}: Panels (a) and (b) plot the sample analogues of $\bar{E}[\Delta\bar c_i^{*2}]$ (\emph{Climate variation}), $\bar{E}[\Delta\bar w_i^{*2}]$ (\emph{Weather variation}) and $\bar{E}[\Delta\bar c_i^*\Delta \bar w_i^*]$ (\emph{Covariance})}. Panel (c) and (d) present the sample analogues of the signal-to-noise ratio and $\omega_S^{LD}$, respectively.} 
\end{figure}

The two analytical examples provide multiple takeaways. First, the choice of $\tau$ presents a trade-off, even when we solely consider consistency as a criterion. To avoid overlapping averaging windows, $\tau$ can take values from 1 to $T/2$. For both analytical examples we consider, the choice of $\tau$ that maximizes the climate signal-to-noise ratio is in the interior of the domain, specifically $T/3$ in the trend-stationary case and $3T/8$ in the unit-root case. A second important takeaway is that, even if one were to maximize the climate signal-to-noise ratio, the magnitude of this maximum depends on the time horizon of the sample, $T$. In both examples, the climate variation is increasing in $T$, whereas the weather variation does not depend on it. For the trend-stationary case, the maximum $SNR_c$ is cubic in $T$, whereas in the unit-root case it is quadratic in $T$. 

We next compute the components of $SNR_c$ and $\omega_S^{LD}$ using the temperature data used in our simulation design in Section \ref{sec:simulations}. Panels (a) and (b) of Figure \ref{fig:sim_SNR} plot the sample analogue of the LD-transformed climate variation ($\bar{E}[\Delta\bar c_i^{*2}]$), the LD-transformed weather variation ($\bar{E}[\Delta\bar w_i^{*2}]$) and their covariance ($\bar{E}[\Delta\bar c_i^*\Delta \bar w_i^*]$), where $c_{it}^*$ is a 10- and 30-year normal, respectively. Panel (a) demonstrates that, similar to the analytical examples, we see that a larger $\tau$ decreases the short-run weather variation. The climate variability also declines with the increase in $\tau$, as in the unit-root case, but only marginally relative to short-run weather variability. The covariance term is negative in this example and gets smaller in magnitude as $\tau$ increases. By contrast, when we consider the 30-year climate normal in Panel (b), we find that $\tau$ has hardly any impact on the climate variation and the covariance components, which are small in magnitude, though it vastly reduces the short-run weather variability.

Finally, we examine the implied $SNR_c$ and $\omega_S^{LD}$ in the 10- and 30-year climate normal case. We first note that the $SNR_c$ is more nonlinear than in the analytical examples (Panel (c) of Figure \ref{fig:sim_SNR}. It increases in $\tau$ up to $\tau=10$, and then decreases and reaches its maximum at $T=12$. For the 30-year normal, however, regardless of the choice of $\tau$, the signal-to-noise ratio is quite low and therefore $\omega_S^{LD}$ is close to one for most values of $\tau\in\{1,\dots,12\}$. As a result, regardless of the choice of $\tau$, the bias of the LD estimator toward the short-run response will be substantive. In practice, researchers can estimate $\omega_S^{LD}$ in their setting in order to assess the magnitude of the bias for various models of $c_{it}^*$ that are plausible in their context. We further discuss how such models can be used in Section \ref{sec:implications}.

\section{Implications for empirical practice}\label{sec:implications}

\noindent\textbf{FE and LD estimands are weighted averages of the long-run and short-run response}. This is a consequence of neither the LD nor the FE transformation being able to isolate the desired variation in temperature directly. The FE transformed temperature might still be contamination with climate variation, whereas the LD transformed temperature might contain variation from short-run weather shocks. As a result, comparisons of LD and FE estimands are not consistent for the extent of adaptation in general. In an empirically-calibrated simulation design using temperature data, we find that the LD contamination weight tends to be larger leading the LD estimator to suffer from greater bias in this design.

\bigskip

\noindent\textbf{Caution is warranted when interpreting non-rejections of adaptation tests}. Adaptation tests based on FE and LD estimators are tests by implication. This is a consequence of long-differencing not imposing assumptions on the climate process. Since they do not based on a null hypothesis equivalent to the no-adaptation hypothesis ($\theta_L=\theta_S$), their non-rejection must be interpreted with caution. Corollary \ref{cor:adaptation_test} demonstrates that if the FE and LD contamination weights sum to one, then a test based on a comparison of LD and FE estimators will have trivial power, regardless of the true extent of adaptation ($\theta_L-\theta_S$).

\bigskip

\noindent\textbf{Choice of $\tau$ may exacerbate the bias of LD estimators}. The analytical and numerical examples demonstrate that the choice of averaging window in the LD estimator ($\tau$) can have substantive implications for the bias of the LD estimator and the power of the adaptation test. They also demonstrate that the averaging window that minimizes the bias of the LD estimator depends on the underlying climate data-generating process. For some data-generating processes, even the LD estimator using the bias-minimizing choice of $\tau$ might still be substantively biased if the climate signal-to-noise ratio is relatively low for that choice of $\tau$. This analysis highlights the importance of considering the time span of a study, the underlying climate data-generating process and the associated climate signal-to-noise ratio when comparing LD and FE estimators and interpreting adaptation tests.

\bigskip

\noindent \textbf{Assessing the magnitude of the contamination weights requires specifying climate models}. Despite the LD estimator not requiring a specification of the climate model, this is necessary to assess its bias in finite samples. Using different models for climate justified by the empirical context, one can estimate both FE and LD contamination weights in order to assess the extent of the bias of FE and LD. Indeed, one can use such estimates to construct confidence intervals for the extent of adaptation, $\theta_L-\theta_S$, by inverting a test of the equality in Corollary \ref{cor:adaptation_test}.\footnote{Once researchers specify climate models, direct estimation of $\theta_L$ and $\theta_S$ using the assumed climate model becomes a compelling alternative as proposed in, for example, \citet{bento2023jeem}. } The test inversion is required in this context, since weak identification concerns would arise when the contamination weights sum to one.

\bigskip

\noindent\textbf{While agnostic about the climate-shock process, the validity of long-differencing rests on a separable outcome model}. Our formal analysis demonstrates that the validity of long-differencing and the implied adaptation test requires an outcome model that is separable in the climate and shock components as in Eq.\ \eqref{eq:outcome_model} (in addition to covariates). Economic theory and climate science can be used to justify such an outcome model. The presence of response heterogeneity, a deviation from separability between observables and unobservables, further highlights the importance of considering the assumptions on the outcome model when interpreting long-differencing results.

\bibliography{bib}

@article{auffhammer2022jeem,
  title={Climate Adaptive Response Estimation: Short and long run impacts of climate change on residential electricity and natural gas consumption},
  author={Auffhammer, Maximilian},
  journal={Journal of Environmental Economics and Management},
  volume={114},
  pages={102669},
  year={2022},
  publisher={Elsevier}
}

@article{baylis2025jpub,
  title={Climate and migration in the United States},
  author={Baylis, Patrick and Bharadwaj, Prashant and Mullins, Jamie T and Obradovich, Nick},
  journal={Journal of Public Economics},
  volume={249},
  pages={105446},
  year={2025},
  publisher={Elsevier}
}

@incollection{carleton2024hecc,
  title={Adaptation to climate change},
  author={Carleton, Tamma and Duflo, Esther and Jack, B Kelsey and Zappal{\`a}, Guglielmo},
  booktitle={Handbook of the Economics of Climate Change},
  volume={1},
  number={1},
  pages={143--248},
  year={2024},
  publisher={Elsevier}
}

@article{carter2018arre,
  title={Identifying the economic impacts of climate change on agriculture},
  author={Carter, Colin and Cui, Xiaomeng and Ghanem, Dalia and M{\'e}rel, Pierre},
  journal={Annual Review of Resource Economics},
  volume={10},
  number={1},
  pages={361--380},
  year={2018},
  publisher={Annual Reviews}
}

@misc{cui2025,
  title  = {Adaptation to Climate Change: New Evidence from U.S. Agriculture},
  author = {Cui, Xiaomeng and Xiao, Zimao},
  year   = {2025},
  howpublished = {Working paper},
  url    = {https://www.dropbox.com/scl/fi/5r891rh1d0pl09n7ph1ef/CX_ClimAgLongDiff_update.pdf?rlkey=cfh5g4hemir99gnswmc10t1jg&e=1&st=xj7tznw6&dl=0}
}

@article{bento2023jeem,
  title={A unifying approach to measuring climate change impacts and adaptation},
  author={Bento, Antonio M and Miller, Noah and Mookerjee, Mehreen and Severnini, Edson},
  journal={Journal of Environmental Economics and Management},
  volume={121},
  pages={102843},
  year={2023},
  publisher={Elsevier}
}

@techreport{bilal2024nber,
  title={The macroeconomic impact of climate change: Global vs. local temperature},
  author={Bilal, Adrien and K{\"a}nzig, Diego R},
  year={2024},
  institution={National Bureau of Economic Research}
}

@article{burke2016aejep,
  title={Adaptation to climate change: Evidence from US agriculture},
  author={Burke, Marshall and Emerick, Kyle},
  journal={American Economic Journal: Economic Policy},
  volume={8},
  number={3},
  pages={106--140},
  year={2016},
  publisher={American Economic Association 2014 Broadway, Suite 305, Nashville, TN 37203-2425}
}

@article{butler2013ncc,
  title={Adaptation of US maize to temperature variations},
  author={Butler, Ethan E and Huybers, Peter},
  journal={Nature Climate Change},
  volume={3},
  number={1},
  pages={68--72},
  year={2013},
  publisher={Nature Publishing Group UK London}
}

@article{carleton2017pnas,
  title={Crop-damaging temperatures increase suicide rates in India},
  author={Carleton, Tamma A},
  journal={Proceedings of the national academy of sciences},
  volume={114},
  number={33},
  pages={8746--8751},
  year={2017},
  publisher={National Academy of Sciences}
}

@article{chen2021jde,
  title={Response and adaptation of agriculture to climate change: Evidence from China},
  author={Chen, Shuai and Gong, Binlei},
  journal={Journal of Development Economics},
  volume={148},
  pages={102557},
  year={2021},
  publisher={Elsevier}
}

@article{ghanem2021what,
author = {Dalia Ghanem  and Aaron Smith },
title = {What Are the Benefits of High-Frequency Data for Fixed Effects Panel Models?},
journal = {Journal of the Association of Environmental and Resource Economists},
volume = {8},
number = {2},
pages = {199-234},
year = {2021},
doi = {10.1086/710968},

URL = {https://www.journals.uchicago.edu/doi/abs/10.1086/710968},
eprint = {https://www.journals.uchicago.edu/doi/pdf/10.1086/710968}
}

@article{cui2024on,
title = {On model selection criteria for climate change impact studies},
journal = {Journal of Econometrics},
volume = {239},
number = {1},
pages = {105511},
year = {2024},
note = {Climate Econometrics},
issn = {0304-4076},
doi = {https://doi.org/10.1016/j.jeconom.2023.105511},
url = {https://www.sciencedirect.com/science/article/pii/S0304407623002270},
author = {Xiaomeng Cui and Bulat Gafarov and Dalia Ghanem and Todd Kuffner}
}

@article{dell2012aejm,
  title={Temperature shocks and economic growth: Evidence from the last half century},
  author={Dell, Melissa and Jones, Benjamin F and Olken, Benjamin A},
  journal={American Economic Journal: Macroeconomics},
  volume={4},
  number={3},
  pages={66--95},
  year={2012},
  publisher={American Economic Association}
}

@techreport{deryugina2017nber,
  title={The marginal product of climate},
  author={Deryugina, Tatyana and Hsiang, Solomon},
  year={2017},
  institution={National Bureau of Economic Research}
}

@article{gibbons2019,
  title={Broken or fixed effects?},
  author={Gibbons, Charles E and Su{\'a}rez Serrato, Juan Carlos and Urbancic, Michael B},
  journal={Journal of Econometric Methods},
  volume={8},
  number={1},
  pages={20170002},
  year={2019},
  publisher={De Gruyter}
}

@article{gospodinov2025,
  title={The Economic Impact of Low-and High-Frequency Temperature Changes},
  author={Gospodinov, Nikolay and Gaffney, Ignacio Lopez and Ng, Serena},
  journal={arXiv preprint arXiv:2505.08950},
  year={2025}
}

@article{hasselmann1976tf,
  title={Stochastic climate models part I. Theory},
  author={Hasselmann, Klaus},
  journal={tellus},
  volume={28},
  number={6},
  pages={473--485},
  year={1976},
  publisher={Taylor \& Francis}
}

@article{lehner2023erc,
  title={Origin, importance, and predictive limits of internal climate variability},
  author={Lehner, Flavio and Deser, Clara},
  journal={Environmental Research: Climate},
  volume={2},
  number={2},
  pages={023001},
  year={2023},
  publisher={IOP Publishing}
}

@article{maher2018grl,
  title={ENSO change in climate projections: forced response or internal variability?},
  author={Maher, Nicola and Matei, Daniela and Milinski, Sebastian and Marotzke, Jochem},
  journal={Geophysical Research Letters},
  volume={45},
  number={20},
  pages={11--390},
  year={2018},
  publisher={Wiley Online Library}
}

@article{marotzke2015nature,
  title={Forcing, feedback and internal variability in global temperature trends},
  author={Marotzke, Jochem and Forster, Piers M},
  journal={Nature},
  volume={517},
  number={7536},
  pages={565--570},
  year={2015},
  publisher={Nature Publishing Group UK London}
}

@article{kolstad2020reep,
  title={Estimating the economic impacts of climate change using weather observations},
  author={Kolstad, Charles D and Moore, Frances C},
  journal={Review of Environmental Economics and Policy},
  year={2020},
  publisher={The University of Chicago Press}
}

@techreport{lemoine2018nber,
  title={Estimating the consequences of climate change from variation in weather},
  author={Lemoine, Derek},
  year={2018},
  institution={National Bureau of Economic Research}
}

@techreport{lemoine2025nber,
  title={A guide to climate damages},
  author={Lemoine, Derek and Hausman, Catherine and Shrader, Jeffrey G},
  year={2025},
  institution={National Bureau of Economic Research}
}

@article{liu2023aejep,
  title={Climate change and labor reallocation: Evidence from six decades of the Indian Census},
  author={Liu, Maggie and Shamdasani, Yogita and Taraz, Vis},
  journal={American Economic Journal: Economic Policy},
  volume={15},
  number={2},
  pages={395--423},
  year={2023},
  publisher={American Economic Association 2014 Broadway, Suite 305, Nashville, TN 37203-2425}
}

@article{merel2021ajae,
  title={Climate econometrics: Can the panel approach account for long-run adaptation?},
  author={M{\'e}rel, Pierre and Gammans, Matthew},
  journal={American Journal of Agricultural Economics},
  volume={103},
  number={4},
  pages={1207--1238},
  year={2021},
  publisher={Wiley Online Library}
}

@article{gammans2017,
  title={Negative impacts of climate change on cereal yields: statistical evidence from France},
  author={Gammans, Matthew and M{\'e}rel, Pierre and Ortiz-Bobea, Ariel},
  journal={Environmental research letters},
  volume={12},
  number={5},
  pages={054007},
  year={2017},
  publisher={IOP Publishing}
}

@article{moore2014ncc,
  title={Adaptation potential of European agriculture in response to climate change},
  author={Moore, Frances C and Lobell, David B},
  journal={Nature Climate Change},
  volume={4},
  number={7},
  pages={610--614},
  year={2014},
  publisher={Nature Publishing Group UK London}
}

@techreport{obolensky2024nber,
  title={Migration, Climate Similarity, and the Consequences of Climate Mismatch},
  author={Obolensky, Marguerite and Tabellini, Marco and Taylor, Charles},
  year={2024},
  institution={National Bureau of Economic Research}
}

@article{obradovich2018pnas,
  title={Empirical evidence of mental health risks posed by climate change},
  author={Obradovich, Nick and Migliorini, Robyn and Paulus, Martin P and Rahwan, Iyad},
  journal={Proceedings of the National Academy of Sciences},
  volume={115},
  number={43},
  pages={10953--10958},
  year={2018},
  publisher={National Academy of Sciences}
}

@article{ponticelli2023nber,
  title={Temperature, adaptation, and local industry concentration},
  author={Ponticelli, Jacopo and Xu, Qiping and Zeume, Stefan},
  journal={NBER Working Paper},
  number={w31533},
  year={2023}
}

@article{schlenker2009pnas,
  title={Nonlinear temperature effects indicate severe damages to US crop yields under climate change},
  author={Schlenker, Wolfram and Roberts, Michael J},
  journal={Proceedings of the National Academy of sciences},
  volume={106},
  number={37},
  pages={15594--15598},
  year={2009},
  publisher={National Acad Sciences}
}

@article{shrader2023wp,
  title={Improving climate damage estimates by accounting for adaptation},
  author={Shrader, Jeffrey},
  journal={Available at SSRN 3212073},
  year={2023}
}

@article{taylor2026jaere,
  title={Irrigation and climate change: Long-run adaptation and its externalities},
  author={Taylor, Charles A},
  journal={Journal of the Association of Environmental and Resource Economists},
  volume={13},
  number={4},
  pages={935--973},
  year={2026},
  publisher={The University of Chicago Press Chicago, IL}
}

@article{won2024ajae,
  title={Understanding the effect of cover crop use on prevented planting losses},
  author={Won, Sunjae and Rejesus, Roderick M and Goodwin, Barry K and Aglasan, Serkan},
  journal={American Journal of Agricultural Economics},
  volume={106},
  number={2},
  pages={659--683},
  year={2024},
  publisher={Wiley Online Library}
}

@article{yu2021sr,
  title={Maladaptation of US corn and soybeans to a changing climate},
  author={Yu, Chengzheng and Miao, Ruiqing and Khanna, Madhu},
  journal={Scientific reports},
  volume={11},
  number={1},
  pages={12351},
  year={2021},
  publisher={Nature Publishing Group UK London}
}

@article{griliches1986errors,
  title={Errors in variables in panel data},
  author={Griliches, Zvi and Hausman, Jerry A},
  journal={Journal of econometrics},
  volume={31},
  number={1},
  pages={93--118},
  year={1986},
  publisher={Elsevier}
}
\newpage
\pagenumbering{gobble}
\appendix
\setcounter{page}{0} 
\setcounter{section}{0}
\begin{center}
    \huge{Supplementary Appendix}
    
\end{center}

\startcontents[sections]
\printcontents[sections]{l}{1}{\setcounter{tocdepth}{2}}

\pagenumbering{arabic}
\renewcommand{\thefigure}{A\arabic{figure}}
\renewcommand{\thetable}{A\arabic{table}}
\setcounter{figure}{0}
\setcounter{table}{0}
\renewcommand{\thefigure}{B\arabic{figure}}
\renewcommand{\thetable}{B\arabic{table}}
\setcounter{figure}{0}
\setcounter{table}{0}
\setcounter{section}{1}
\appendix
\newpage
\section{Proofs of the results in the main text}
\label{appendix:proofs}
\subsection{Proof of Proposition \ref{prop:FE}}

\begin{proof}
\noindent 
(i)
\begin{eqnarray*}
    \beta_{FE}&=&\frac{\bar{E}[\sum_t \tilde{x}_{it}(\theta_S \tilde w^*_{it} + \theta_L\tilde c^*_{it} + \tilde \varepsilon_{it})]}{\bar{E}[\sum_t \tilde{x}_{it}^2]}= \theta_S \frac{\bar{E}[\sum_t  \tilde{x}_{it}\tilde w^*_{it}]}{\bar{E}[\sum_t  \tilde{x}_{it}^2]} + \theta_L \frac{\bar{E}[\sum_t  \tilde{x}_{it}\tilde c^*_{it}]}{\bar{E}[\sum_t \tilde{x}_{it}^2]} \\
    &=& \theta_S \frac{\bar{E}[\sum_t (\tilde w^*_{it} + \tilde c^*_{it})\tilde w^*_{it}]}{\bar{E}[\sum_t  \tilde{x}_{it}^2]} + \theta_L \frac{\bar{E}[\sum_t  (\tilde w^*_{it} + \tilde c^*_{it})\tilde c^*_{it}]}{\bar{E}[\sum_t \tilde{x}_{it}^2]} \\
    &=&\theta_S\frac{\bar{E}[\sum_t\tilde{w}^{*2}_{it}]+\bar{E}[\sum_t\tilde{w}^{*}_{it}\tilde{c}^{*}_{it}]}{\bar{E}[\sum_t\tilde{x}_{it}^2]}+\theta_L\frac{\bar{E}[\sum_t\tilde{c}^{*2}_{it}]+\bar{E}[\sum_t\tilde{w}^{*}_{it}\tilde{c}^{*}_{it}]}{\bar{E}[\sum_t\tilde{x}_{it}^2]},
\end{eqnarray*}
where the first equality follows by plugging the within-group demeaned version of \eqref{eq:outcome_model} for $\tilde{y}_{it}$ and invoking strict exogeneity. The second equality follows from noting that $\tilde{x}_{it}=\tilde c_{it}^*+\tilde{w}_{it}^*$.

The result follows from the definition of $\omega_S^{FE}$ and $\omega_L^{FE}$ in Proposition \ref{prop:FE}. The summing-to-one result is immediate from noting that the sum of the numerators of $\omega_S^{FE}$ and $\omega_L^{FE}$ equals to their common denominator.

\noindent (ii)
Subtracting $\theta_S$ from both sides of the previous equality yields
\begin{eqnarray*}
\beta_{FE}-\theta_S&=&\theta_S\frac{\bar{E}[\sum_t\tilde{w}^{*2}_{it}]}{\bar{E}[\sum_t\tilde{x}_{it}^2]}+\theta_L\frac{\bar{E}[\sum_t\tilde{c}^{*2}_{it}]}{\bar{E}[\sum_t\tilde{x}_{it}^2]}+(\theta_S+\theta_L)\frac{\bar{E}[\sum_t\tilde{w}^{*}_{it}\tilde{c}^{*}_{it}]}{\bar{E}[\sum_t\tilde{x}_{it}^2]}-\theta_S\\
&=&\theta_S\left(\frac{\bar{E}[\sum_t\tilde{w}^{*2}_{it}]+\bar{E}[\sum_t\tilde{w}^{*}_{it}\tilde{c}^{*}_{it}]-\bar{E}[\sum_t\tilde{x}_{it}^2]}{\bar{E}[\sum_t\tilde{x}_{it}^2]}\right)+\theta_L\left(\frac{\bar{E}[\sum_t\tilde{c}^{*2}_{it}]+\bar{E}[\sum_t\tilde{w}^{*}_{it}\tilde{c}^{*}_{it}]}{\bar{E}[\sum_t\tilde{x}_{it}^2]}\right)
\end{eqnarray*}
Since $\tilde{x}_{it}=\tilde{c}^{*}_{it}+\tilde{w}^*_{it}$, it follows that $\bar{E}\left[\sum_t \tilde{x}_{it}^2\right]=\bar{E}\left[\sum_t\tilde{c}^{*2}_{it}\right]+\bar{E}\left[\sum_t\tilde{w}^{*2}_{it}\right]+2\bar{E}\left[\sum_t\tilde{w}^{*}_{it}\tilde{c}^{*}_{it}\right]$. Plugging this equality in the parentheses in the first term on the RHS of the last equality yields
\begin{eqnarray*}
\beta_{FE}-\theta_S&=&(\theta_L-\theta_S)\left(\frac{\bar{E}[\sum_t\tilde{c}^{*2}_{it}]+\bar{E}[\sum_t\tilde{w}^{*}_{it}\tilde{c}^{*}_{it}]}{\bar{E}[\sum_t\tilde{x}_{it}^2]}\right)
\end{eqnarray*}
\end{proof}

\subsection{Proof of Proposition \ref{prop:LD}}

\noindent (i) The proof follows by similar steps as in the proof of Proposition \ref{prop:FE}.
\begin{eqnarray*}
    \beta_{LD}&\equiv & \frac{\bar{E}[  \Delta \bar{x}_{i}\Delta\bar{y}_{i}]}{\bar{E}[\Delta \bar{x}_{i}^2]} =\frac{\bar{E}[ \Delta \bar{x}_{i}(\theta_S \Delta \bar w^*_{i} + \theta_L \Delta \bar c^*_{i} + \Delta \bar \varepsilon_{i})]}{\bar{E}[\Delta \bar{x}_{i}^2]} \\
      &=& \theta_S \frac{\bar{E}[ \Delta \bar{x}_{i}\Delta \bar w^*_{i}]}{\bar{E}[ \Delta \bar{x}_{i}^2]} + \theta_L \frac{\bar{E}[\Delta \bar{x}_{i}\Delta \bar c^*_{i}]}{\bar{E}[ \Delta \bar{x}_{i}^2]} \\
    &=& \theta_S \frac{\bar{E}[ (\Delta \bar w^*_{i} + \Delta \bar c^*_{i})\Delta \bar w^*_{i}]}{\bar{E}[ \Delta \bar{x}_{i}^2]} + \theta_L \frac{\bar{E}[ (\Delta \bar w_{i}^{*} + \Delta \bar c^*_{i})\Delta \bar c^*_{i}]}{\bar{E}[ \Delta \bar{x}_{i}^2]} \\
     &=& \theta_S \frac{\bar{E}[\Delta \bar w_{i}^{*2}]+\bar{E}[ \Delta \bar w_{i}^{*}\Delta \bar c^*_{i}]}{\bar{E}[ \Delta  \bar{x}_{i}^2]} + \theta_L \frac{\bar{E}[ \Delta \bar c_{i}^{*2}]+\bar{E}[ \Delta \bar w_{i}^{*}\Delta \bar c^*_{i}]}{\bar{E}[ \Delta \bar{x}_{i}^2]}.  
\end{eqnarray*}
The result follows from the definition of $\omega_S^{LD}$ and $\omega_L^{LD}$ in Proposition \ref{prop:LD}. The summing-to-one result is immediate from noting that the sum of the numerators of $\omega_S^{LD}$ and $\omega_L^{LD}$ equals to their common denominator.

(ii)
Subtracting $\theta_L$ from both sides of the result in (i) yields
\begin{eqnarray*}
\beta_{LD}-\theta_L
&=&\theta_S\left(\frac{\bar{E}[\Delta \bar{w}_{i}^{*2}]}{\bar{E}[\Delta \bar{x}_{i}^2]}+\frac{\bar{E}[\Delta \bar{w}_{i}^{*} \Delta \bar{c}_{i}^{*}]}{\bar{E}[\Delta \bar{x}_{i}^2]}\right)+\theta_L\left(\frac{\bar{E}[\Delta \bar{c}_{i}^{*2}]}{\bar{E}[\Delta \bar{x}_{i}^2]}+\frac{\bar{E}[\Delta \bar{w}_{i}^{*}\Delta \bar{c}_{i}^{*}]}{\bar{E}[\Delta \bar{x}_{i}^2]} - 1\right)
\end{eqnarray*}

\noindent Simplifying the last equality similar to the last step in the proof of Proposition \ref{prop:FE} yields
\begin{eqnarray*}
\beta_{LD}-\theta_L&=&(\theta_S-\theta_L)\left(\frac{\bar{E}\left[\Delta \bar{w}^{*2}_{i}\right] + \bar{E}\left[\Delta\bar{w}^{*}_{i}\Delta\bar{c}^{*}_{i}\right]}{\bar{E}[\Delta\bar{x}^{2}_i]}\right)
\end{eqnarray*}

\qed

\subsection{Proof of Corollary \ref{cor:adaptation_test}}

\begin{proof}
Subtracting $\beta_{FE}$ from $\beta_{LD}$ together with invoking $\omega_L^{FE}+\omega_S^{FE}=1$ and $\omega_S^{LD}+\omega_L^{LD}=1$ from Propositions \ref{prop:FE} and \ref{prop:LD}, respectively, yields the result:
\begin{eqnarray*}
    \beta_{LD} - \beta_{FE} &=& (\omega^{LD}_S \theta_S + \omega^{LD}_L \theta_L) - (\omega^{FE}_S \theta_S + \omega^{FE}_L \theta_L) \\
    &=& ((1-\omega^{LD}_L) \theta_S + \omega^{LD}_L \theta_L) - ((1-\omega^{FE}_L) \theta_S + \omega^{FE}_L \theta_L) \\
    &=& (\theta_L - \theta_S)(1-\omega^{LD}_S - \omega^{FE}_L)
\end{eqnarray*}
where
\begin{eqnarray*}
    \omega_S^{LD} \equiv \left(\frac{\bar{E}\left[\Delta \bar{w}^{*2}_{i}\right] + \bar{E}\left[\Delta\bar{w}^{*}_{i}\Delta\bar{c}^{*}_{i}\right]}{\bar{E}[\Delta\bar{x}^{2}_i]}\right), \qquad \omega_L^{FE} \equiv \left(\frac{\bar{E}[\sum_t\tilde{c}^{*2}_{it}]+\bar{E}[\sum_t\tilde{w}^{*}_{it}\tilde{c}^{*}_{it}]}{\bar{E}[\sum_t\tilde{x}_{it}^2]}\right).
\end{eqnarray*}
\end{proof}
\section{Choice of $\tau$}\label{app:choice_tau}
\subsection{Linearly Trending Case}
\begin{proposition}[Linearly Trending Case] \label{prop:tau_for_trend}

Suppose that for $i=1,\dots,N$ and $t=1,\dots, T$ $c^*_{it} = \mu_i + \delta_i t$, $w^*_{it} \overset{i.i.d.}{\sim} (0, \sigma_w^2)$ and $\delta_i\overset{i.i.d.}{\sim}(0,\sigma_\delta^2)$, $w_{it}^* \perp(\mu_i,\delta_i)$, where $\sigma_w^2>0$, $\sigma_\delta^2>0$ and $Cov(\mu_i,\delta_i)=0$. Then,
 \begin{enumerate}[(i)]\item $SNR_c=\tau(T-\tau)^2\frac{\sigma_\delta^2}{2\sigma_w^2}$ and $\omega_S^{LD} = \dfrac{1}{1 + \left[\tau(T-\tau)^2 \frac{\bar \sigma_{\delta}^2}{2\bar \sigma_{w}^2} \right]}$
 \item $\omega_S^{LD}$ is minimized when $\tau^* =T/3$.
\end{enumerate}
\end{proposition}

\begin{proof}
\textit{(i)} We first derive the components of $\omega_S^{LD}$ from
Proposition \ref{prop:LD} under the maintained assumptions.

\begin{eqnarray*}
    \bar{E}[\Delta \bar w_i^{*2}]=Var\left(\frac{1}{\tau}\sum_{t=T-\tau+1}^Tw_{it}-\frac{1}{\tau}\sum_{t=1}^{\tau}w_{it}\right)=\frac{2\sigma_w^2}{\tau},
    \label{eq:vw}
\end{eqnarray*}
where the first equality follows from $E[w_i^{*}]=0$ under the maintained assumptions. The second equality follows from the definition of $\Delta \bar w_i^{*}$, whereas the last equality follows from the i.i.d.\ assumption imposed on $w_{it}^*$ and $Var(w_{it}^*)=\sigma_w^2$.
\begin{eqnarray*}\bar{E}[\Delta\bar{c}_i^*]&=&\bar{E}\left[\left(\frac{1}{\tau}\sum_{t=T-\tau+1}^T(\mu_i+\delta_i t)-\frac{1}{\tau}\sum_{t=1}^\tau(\mu_i+\delta_i t) \right)^2\right]\\
&=&\bar{E}\left[\left(\delta_i\frac{1}{\tau}\left(\frac{\tau(T-\tau+1+T)}{2}-\frac{\tau(\tau+1)}{2}\right)\right)^2\right]=\bar{E}\left[\left(\delta_i\left(\frac{T-\tau+1+T-(\tau+1)}{2}\right)\right)^2\right]\\
&=&\bar{E}\left[\left(\delta_i\left(\frac{2T-2\tau}{2}\right)\right)^2\right]\equiv \sigma_\delta^2(T-\tau)^2,
    \label{eq:vc}
\end{eqnarray*}
where the first equality follows from the long-difference transformation differencing out $\mu_i$. The last equality follows from $\delta_i\overset{i.i.d.}{\sim} (0,\sigma^2)$. The remaining equalities follow from standard algebraic manipulations.

Since $\Delta \bar w_{i}^*$ and $\Delta \bar c_{i}^*$ are independent, zero-mean random variables under the maintained assumptions, $\bar{E}[\Delta \bar w_{i}^*\Delta \bar c_{i}^*]=0$. As a result, $\omega_S^{LD}$ simplifies to
\begin{eqnarray}
\omega_S^{LD}& = &
= \frac{\dfrac{2\sigma_w^2}{\tau}}{\dfrac{2\sigma_w^2}{\tau} +
\sigma_\delta^2(T-\tau)^2}=\frac{1}{1+\frac{\tau(T-\tau)^2}{2}\frac{\sigma_\delta^2}{\sigma_w^2}}.
\label{eq:omega_simplified}
\end{eqnarray}

\end{proof}
\subsection{Unit Root Case}
\begin{proposition}[Unit Root Case]
\label{prop:tau_unit_root}
Suppose that $x_{it} = x_{i,t-1} + \eta_{it}$ with $c^*_{it} = x_{i,t-1}$ and
$w^*_{it} = \eta_{it}$, where $\eta_{it} \overset{i.i.d.}{\sim} (0,\sigma_\eta^2)$, $\sigma_\eta^2>0$,
and the process is initialized at $x_{i0}$ for each $i$. For $1 \le \tau \le T/2$,
\begin{enumerate}[(i)]
\item $SNR_c=T\tau - \frac{4\tau^2}{3} + \frac{\tau}{3}-1$ and $\omega_S^{LD} = \dfrac{3}{3\tau T - 4\tau^2 + 1}$;
\item $\omega_S^{LD}$ is minimized at $\tau^* = 3T/8$.
\end{enumerate}
\end{proposition}

\begin{proof}
\textit{(i)} We first derive all components of $\omega_S^{LD}$ defined in Proposition \ref{prop:LD} under the maintained assumptions. To do so, it is first helpful to conduct recursive substitution in $x_{it}$ up to the initial condition $x_{i0}$, such that $x_{it} = x_{i0} + \sum_{s=1}^t \eta_{is}$, and subsequently $c_{it}^*=x_{i0}+\sum_{s=1}^{t-1}\eta_{is}$.

Since the windows are disjoint ($\tau \le T/2$) and $\eta_{it}$ is i.i.d. with $E[\eta_{it}]=0$, it follows that
\begin{equation*}
\bar{E}[\Delta \bar w_i^{*2}]= Var\left(\frac{1}{\tau}\sum_{t=T-\tau+1}^T\eta_{it}
- \frac{1}{\tau}\sum_{t=1}^\tau\eta_{it}\right) = \frac{2\sigma_\eta^2}{\tau}.
\end{equation*}

Next, we express $\Delta\bar c_i^*$ as a linear combination of the innovations. First, note that since $c^*_{it}= x_{i0}+1\{t>1\}\sum_{s=1}^{t-1}\eta_{is}$, $\Delta\bar c_i^*$ can be simplified as follows, where the third equality follows from the sums involving $x_{i0}$ canceling out.

\begin{eqnarray}\Delta\bar c_i^*&\equiv&\frac{1}{\tau}\sum_{t=T-\tau+1}^T c_{it}^*-\frac{1}{\tau}\sum_{t=1}^\tau c_{it}^*=\frac{1}{\tau}\sum_{t=T-\tau+1}^T\left(x_{i0}+\sum_{s=1}^{t-1}\eta_{is}\right)-\frac{1}{\tau}\sum_{t=1}^\tau\left(x_{i0}+1\{t>1\}\sum_{s=1}^{t-1}\eta_{is}\right)\nonumber\\
&=&\frac{1}{\tau}\sum_{t=T-\tau+1}^T \sum_{s=1}^{t-1}\eta_{is}-\frac{1}{\tau}\sum_{t=2}^\tau\sum_{s=1}^{t-1}\eta_{is}=\sum_{s=1}^{T-\tau}\eta_{is}+\frac{1}{\tau}\sum_{s=T-\tau+1}^T(T-s)\eta_{is}-\frac{1}{\tau}\sum_{s=1}^{\tau-1}(\tau-s)\eta_{is}\nonumber\\
&=&\sum_{s=1}^{\tau-1}\frac{s}{\tau}\eta_{is}+\sum_{s=\tau}^{T-\tau}\eta_{is}+\sum_{s=T-\tau+1}^T\frac{(T-s)}{\tau}\eta_{is}\label{eq:LD_cit}\end{eqnarray}

As a result, Eq.\ \eqref{eq:LD_cit} implies that $\Delta\bar c_i^* = \sum_{s=1}^T \gamma_s\,\eta_{is}$
with
\begin{equation}
\gamma_s = \begin{cases}
s/\tau, & 1 \le s \le \tau,\\
1, & \tau < s \le T-\tau,\\
(T-s)/\tau, & T-\tau < s \le T.
\end{cases}
\label{eq:rw_coefs}
\end{equation}
Under the i.i.d.\ assumption imposed on $\eta_{it}$ with $E[\eta_{it}]=0$, $\bar{E}[\Delta \bar c_i^{*2}]$ simplifies as follows
\begin{eqnarray*}
\bar{E}[\Delta \bar c_i^{*2}]&=&Var(\Delta \bar c_i^*)=\sigma_\eta^2\left[\frac{1}{\tau^2}\sum_{s=1}^{\tau}s^2 + (T-2\tau)
+ \frac{1}{\tau^2}\sum_{s=T-\tau+1}^T(T-s)^2\right]\\
&=& \sigma_\eta^2\left[\frac{(\tau+1)(2\tau+1)}{6\tau} + (T-2\tau)
+ \frac{(\tau-1)(2\tau-1)}{6\tau}\right]\\
&=& \sigma_\eta^2\left[(T-2\tau) + \frac{(\tau+1)(2\tau+1)+(\tau-1)(2\tau-1)}{6\tau}\right]\\
&=&\sigma_\eta^2\left[(T-2\tau)-\frac{4\tau^2+2}{6\tau}\right]=\sigma_\eta^2\left[(T-2\tau)+\frac{2\tau^2+1}{3\tau}\right]\\
&=& \sigma_\eta^2\left(T - \frac{4\tau}{3} + \frac{1}{3\tau}\right)
\end{eqnarray*}
where the second equality follows from $\sum_{s=1}^\tau s^2=\frac{\tau(\tau+1)(2\tau+1)}{6}$ and $\sum_{s=T-\tau+1}^T(T-s)^2=\sum_{r=0}^{\tau-1}r^2=\frac{(\tau-1)\tau(2(\tau-1)+1)}{6}=\frac{(\tau-1)\tau(2\tau-1)}{6}$. The remaining equalities follow by standard algebraic manipulation.

Since $\Delta\bar w_i^* = \sum_{s=1}^T \lambda_s\,\eta_{is}$ with
$\lambda_s = -\frac{1}{\tau}$ for $s=1,\dots,\tau$, $\frac{1}{\tau}$ for $s=T-\tau+1,\dots,T$, and
$0$ otherwise, using the definition of $\gamma_s$ in Eq.\ \eqref{eq:rw_coefs} and invoking the i.i.d.\ assumption imposed on $\eta_{it}$ with $E[\eta_{it}]=0$ yields
\begin{equation*}\bar{E}[\Delta\bar c_i^*\Delta \bar w_i^*] = \sigma_\eta^2\sum_{s=1}^T \lambda_s\gamma_s
= \sigma_\eta^2\left[-\frac{1}{\tau^2}\sum_{s=1}^{\tau}s
+ \frac{1}{\tau^2}\sum_{s\in L}(T-s)\right]
= \sigma_\eta^2\left[-\frac{\tau+1}{2\tau} + \frac{\tau-1}{2\tau}\right]
= -\frac{\sigma_\eta^2}{\tau}.
\end{equation*}

Substituting into the definition of $\omega_S^{LD}$ in Proposition \ref{prop:LD} noting that $\bar{E}[\Delta \bar w_i^{*2}]+2\bar{E}[\Delta \bar w_i^*\Delta\bar c_i^*]=0$, 
so the denominator collapses to $\bar{E}[\Delta\bar c_i^{*2}]$, while
the numerator equals $\sigma_\eta^2/\tau$. Hence,
\begin{equation*}
\omega_S^{LD}
= \frac{\sigma_\eta^2/\tau}{\sigma_\eta^2\left(T - \frac{4\tau}{3} + \frac{1}{3\tau}\right)}
= \frac{3}{3\tau T - 4\tau^2 + 1}.
\end{equation*}
and
\begin{eqnarray}
    SNR_c&=&\frac{\bar{E}[\Delta\bar c_i^{*2}]+\bar{E}[\Delta \bar c_i^*\Delta\bar w_i^*]}{\bar{E}[\Delta\bar w_i^{*2}]+\bar{E}[\Delta \bar c_i^*\Delta\bar w_i^*]}=T\tau - \frac{4\tau^2}{3} + \frac{\tau}{3}-1.
\end{eqnarray}
\textit{(ii)} Minimizing $\omega_S^{LD}$ is equivalent to maximizing
$q(\tau) \equiv 3\tau T - 4\tau^2$, a strictly concave parabola maximized at
$\tau^* = \frac{3T}{8}$, which implies the result.
\end{proof}
\section{FE and LD probability limits under response heterogeneity}\label{app:heterogeneity}
\begin{proposition}\label{prop:FE_heterogeneity}
    Suppose that $\bar{E}[\sum_t\tilde{x}_{it}\tilde{y}_{it}]<\infty$ and $0<\bar{E}[\sum_t\tilde{x}_{it}^2]<\infty$. Suppose further that $y_{it}=\theta_{S,i}w_{it}^*+\theta_{L,i}c_{it}^*+\alpha_i+\varepsilon_{it}$ and $E[\varepsilon_{it}|X_i] = 0$ for all $i=1,\dots,n$ and $t=1,\dots,T$. Then,
 \begin{eqnarray*}\beta_{FE}&=&\bar{E}\left[\theta_{S,i}\frac{E[\sum_t\tilde w^{*2}_{it}|\theta_{S,i}]+E[\sum_t\tilde w^{*}_{it}\tilde{c}_{it}^*|\theta_{S,i}]}{\bar{E}[\sum_t \tilde{x}_{it}^2]}\right]+\bar{E}\left[\theta_{L,i}\frac{E[\sum_t\tilde c^{*2}_{it}|\theta_{L,i}]+E[\sum_t\tilde w^{*}_{it}\tilde{c}_{it}^*|\theta_{L,i}]}{\bar{E}[\sum_t \tilde{x}_{it}^2]}\right]
    \end{eqnarray*}
\end{proposition}
\begin{proof}
\noindent 
\begin{eqnarray*}
    \beta_{FE}&\equiv&\frac{\bar{E}[\sum_t\tilde{x}_{it}\tilde{y}_{it}]}{\bar{E}[\sum_t\tilde{x}_{it}^2]}= \frac{\bar{E}[\sum_t \tilde{x}_{it}(\theta_{S,i} \tilde w^*_{it} + \theta_{L,i}\tilde c^*_{it} + \tilde \varepsilon_{it})]}{\bar{E}[\sum_t \tilde{x}_{it}^2]} \\
    &=&\frac{\bar{E}[\theta_{S,i}\sum_t \tilde{x}_{it} \tilde w^*_{it}]}{\bar{E}[\sum_t \tilde{x}_{it}^2]} +\frac{\bar{E}[\theta_{L,i}\sum_t \tilde{x}_{it}\tilde c^*_{it}]}{\bar{E}[\sum_t \tilde{x}_{it}^2]} \\
    &=&\lim_{n\rightarrow\infty}\frac{1}{n}\sum_{i=1}^nE\left[\theta_{S,i}\frac{E[\sum_t\tilde w^{*2}_{it}|\theta_{S,i}]+E[\sum_t\tilde w^{*}_{it}\tilde{c}_{it}^*|\theta_{S,i}]}{\bar{E}[\sum_t \tilde{x}_{it}^2]}\right]\\&&+\lim_{n\rightarrow\infty}\frac{1}{n}\sum_{i=1}^n E\left[\theta_{L,i}\frac{E[\sum_t\tilde c^{*2}_{it}|\theta_{L,i}]+E[\sum_t\tilde w^{*}_{it}\tilde{c}_{it}^*|\theta_{L,i}]}{\bar{E}[\sum_t \tilde{x}_{it}^2]}\right]\\
&=&\bar{E}\left[\theta_{S,i}\frac{E[\sum_t\tilde w^{*2}_{it}|\theta_{S,i}]+E[\sum_t\tilde w^{*}_{it}\tilde{c}_{it}^*|\theta_{S,i}]}{\bar{E}[\sum_t \tilde{x}_{it}^2]}\right]+\bar{E}\left[\theta_{L,i}\frac{E[\sum_t\tilde c^{*2}_{it}|\theta_{L,i}]+E[\sum_t\tilde w^{*}_{it}\tilde{c}_{it}^*|\theta_{L,i}]}{\bar{E}[\sum_t \tilde{x}_{it}^2]}\right]
\end{eqnarray*}
where the penultimate equality follows from the definition of $\bar{E}[\cdot]$ and applying the law of iterated expectations inside the unconditional expectation for each $i$.

\end{proof}
\begin{proposition}\label{prop:LD_heterogeneity}
    Suppose that $\bar{E}[\Delta \bar x_i\Delta \bar y_i]<\infty$ and $0<\bar{E}[\Delta \bar x_i^2]<\infty$. Suppose further that $y_{it}=\theta_{S,i}w_{it}^*+\theta_{L,i}c_{it}^*+\alpha_i+\varepsilon_{it}$ and $E[\varepsilon_{it}|X_i] = 0$ for all $i=1,\dots,n$ and $t=1,\dots,T$.  Then,
 \begin{eqnarray*}\beta_{LD}&=& \bar{E}\left[\theta_{L,i}\frac{E[\Delta \bar{c}_i^{*2}|\theta_{L,i}]+E[\Delta\bar{c}_i\Delta\bar{w}_i|\theta_{L,i}]}{\bar{E}[\Delta \bar{x}_{i}^2]}\right]+ \bar{E}\left[\theta_{S,i}\frac{E[\Delta \bar{w}_i^{*2}|\theta_{S,i}]+E[\Delta\bar{c}_i\Delta\bar{w}_i|\theta_{S,i}]}{\bar{E}[\Delta \bar{x}_{i}^2]}\right]
    \end{eqnarray*}
    Suppose further that $E[\Delta \bar x_i^2|\theta_{L,i},\theta_{S,i}]>0$.
    \begin{eqnarray*}
    \beta_{LD}&=&\bar{E}\left[\frac{E[\Delta \bar x_i^2|\theta_{L,i},\theta_{S,i}]}{\bar{E}[\Delta \bar{x}_{i}^2]}\theta_{L,i}\right]+\bar{E}\left[\frac{E[\Delta \bar x_i^2|\theta_{L,i},\theta_{S,i}]}{\bar{E}[\Delta \bar{x}_{i}^2]}\omega_{S,i}^{LD}(\theta_{S,i}-\theta_{L,i})\right],\end{eqnarray*}
    where $\omega_{S,i}^{LD}\equiv\frac{E[\Delta \bar{w}_i^{*2}|\theta_{L,i},\theta_{S,i}]+E[\Delta\bar{c}_i\Delta\bar{w}_i|\theta_{L,i},\theta_{S,i}]}{E[\Delta \bar{x}_{i}^2|\theta_{L,i},\theta_{S,i}]}$.
\end{proposition}
\begin{proof} The result follows by similar steps as those made in the proof of Proposition \ref{prop:FE_heterogeneity}, 
\begin{eqnarray*}
    \beta_{LD}&\equiv & \frac{\bar{E}[  \Delta \bar{x}_{i}\Delta\bar{y}_{i}]}{\bar{E}[\Delta \bar{x}_{i}^2]}
     =\frac{\bar{E}[ \Delta \bar{x}_{i}(\theta_{S,i} \Delta \bar w^*_{i} + \theta_{L,i} \Delta \bar c^*_{i} + \Delta \bar \varepsilon_{i})]}{\bar{E}[\Delta \bar{x}_{i}^2]}\\
&=& \bar{E}\left[\theta_{L,i}\frac{E[\Delta \bar{c}_i^{*2}|\theta_{L,i},\theta_{S,i}]+E[\Delta\bar{c}_i\Delta\bar{w}_i|\theta_{L,i},\theta_{S,i}]}{\bar{E}[\Delta \bar{x}_{i}^2]}\right]\\&&+ \bar{E}\left[\theta_{S,i}\frac{E[\Delta \bar{w}_i^{*2}|\theta_{L,i},\theta_{S,i}]+E[\Delta\bar{c}_i\Delta\bar{w}_i|\theta_{L,i},\theta_{S,i}]}{\bar{E}[\Delta \bar{x}_{i}^2]}\right]\\
&=&\bar{E}\left[\theta_{L,i}\frac{E[\Delta \bar x_i^2|\theta_{L,i},\theta_{S,i}]}{\bar{E}[\Delta \bar{x}_{i}^2]}\frac{E[\Delta \bar{c}_i^{*2}|\theta_{L,i},\theta_{S,i}]+E[\Delta\bar{c}_i\Delta\bar{w}_i|\theta_{L,i},\theta_{S,i}]}{E[\Delta \bar{x}_{i}^2|\theta_{L,i},\theta_{S,i}]}\right] \\
&&+\bar{E}\left[\theta_{S,i}\frac{E[\Delta \bar x_i^2|\theta_{L,i},\theta_{S,i}]}{\bar{E}[\Delta \bar{x}_{i}^2]}\frac{E[\Delta \bar{w}_i^{*2}|\theta_{L,i},\theta_{S,i}]+E[\Delta\bar{c}_i\Delta\bar{w}_i|\theta_{L,i},\theta_{S,i}]}{E[\Delta \bar{x}_{i}^2|\theta_{L,i},\theta_{S,i}]}\right]\\
&=&\bar{E}\left[\frac{E[\Delta \bar x_i^2|\theta_{L,i},\theta_{S,i}]}{\bar{E}[\Delta \bar{x}_{i}^2]}\left(\theta_{L,i}(1-\omega_{S,i}^{LD})+\theta_{S,i}\omega_{S,i}^{LD}\right)\right]\\
&=&\bar{E}\left[\frac{E[\Delta \bar x_i^2|\theta_{L,i},\theta_{S,i}]}{\bar{E}[\Delta \bar{x}_{i}^2]}\theta_{L,i}\right]+\bar{E}\left[\frac{E[\Delta \bar x_i^2|\theta_{L,i},\theta_{S,i}]}{\bar{E}[\Delta \bar{x}_{i}^2]}\omega_{S,i}^{LD}(\theta_{S,i}-\theta_{L,i})\right]
\end{eqnarray*}
where $\omega_{S,i}^{LD}\equiv\frac{E[\Delta \bar{w}_i^{*2}|\theta_{L,i},\theta_{S,i}]+E[\Delta\bar{c}_i\Delta\bar{w}_i|\theta_{L,i},\theta_{S,i}]}{E[\Delta \bar{x}_{i}^2|\theta_{L,i},\theta_{S,i}]}$.
\end{proof}

\renewcommand{\thefigure}{D\arabic{figure}}
\renewcommand{\thetable}{D\arabic{table}}
\setcounter{figure}{0}
\setcounter{table}{0}
\section{Supplement to Section \ref{sec:simulations}}

\subsection{Data description for empirically-calibrated simulation}
\label{app_section:sim1_data}

This section provides documentation of the data used in the empirical illustration. We construct a dataset following the methods outlined in \citet{burke2016aejep} and \citet{schlenker2009pnas}, for our corn yield and weather variables. Our dataset contains the exact same set of counties (N = 1531) used in \citet{burke2016aejep}, but with an additional 20 years of data (1950-2022). It is necessary that our dataset contain additional years for two reasons. First, additional years allow us to establish different choices for $c^*_{it}$ (e.g., linear trend, climate normal). For instance, in order to define the unobserved climate component as a 30-year climate normal, we require a full 30-years before the estimation sample. Second, to compare the decomposed variance weights across various specifications of the LD model (i.e. longer windows ($\tau$) or longer differences ($T$)) additional data is needed. 

Table \ref{table:validate_data_be} presents estimates using panel FE and LD that replicates results from \citet{burke2016aejep} and compares with our constructed data set. 

\subsection{LD and panel data dimensions}
\label{app:ld_explanation}

For the simulations in Section \ref{sec:simulations}, varying sizes of LD averages are compared. The center of each average is fixed with more or less years on either side of the center being used to calculate the average. Figure \ref{fig:year_info} illustrates two examples.

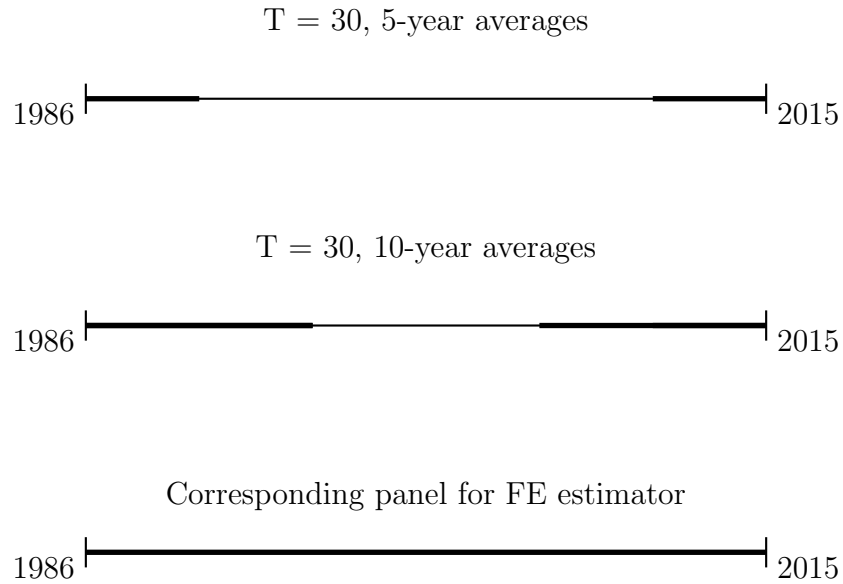
\begin{figure}[H]
\centering
\begin{tikzpicture}[xscale=0.3, yscale=1]

\draw[thick] (0,0) -- (30,0);
% Vertical ticks at ends
\draw[thick] (0,-0.2) -- (0,0.2);
\draw[thick] (30,-0.2) -- (30,0.2);
% Year labels at ends
\node[left] at (0,-0.2) {1986};
\node[right] at (30,-0.2) {2015};
% Highlight segment: 1990--2000 (10 years centered on 1995)
\draw[line width=2pt] (0,0) -- (5,0);
% Highlight segment: 2010--2020 (10 years centered on 2015)
\draw[line width=2pt] (25,0) -- (30,0);
% Title
\node[below,yshift=4pt] at (15,1.2) {T = 30, 5-year averages};

\draw[thick] (0,-3) -- (30,-3);
% Vertical ticks at ends
\draw[thick] (0,-3.2) -- (0,-2.8);
\draw[thick] (30,-3.2) -- (30,-2.8);
% Year labels at ends
\node[left] at (0,-3.2) {1986};
\node[right] at (30,-3.2) {2015};
% Highlight segment: 1990--2000 (10 years centered on 1995)
\draw[line width=2pt] (0,-3) -- (10,-3);
% Highlight segment: 2010--2020 (10 years centered on 2015)
\draw[line width=2pt] (20,-3) -- (30,-3);
% Title
\node[below,yshift=4pt] at (15,-1.8) {T = 30, 10-year averages};

% Highlight segment: 1990--2000 (10 years centered on 1995)
\draw[line width=2pt] (0,-3) -- (10,-3);
% Highlight segment: 2010--2020 (10 years centered on 2015)
\draw[line width=2pt] (25,-3) -- (30,-3);

% PANEL
\draw[line width=2pt] (0,-6) -- (30,-6);
% Vertical ticks at ends
\draw[thick] (0,-6.2) -- (0,-5.8);
\draw[thick] (30,-6.2) -- (30,-5.8);
% Year labels at ends
\node[left] at (0,-6.2) {1986};
\node[right] at (30,-6.2) {2015};
% Top bracket: full panel
\node[above,yshift=6pt] at (15,-5.8) {Corresponding panel for FE estimator};

\end{tikzpicture}
    \caption{Visual description of which years of data are included (bolded) in LD and corresponding FE}
    \label{fig:year_info}
\end{figure}
\begin{table}[H]
   \caption{Comparison of estimates using data from \citet{burke2016aejep} with estimates from constructed dataset}
   \label{table:validate_data_be}
   \vspace{-0.5em}
   \centering
   \footnotesize
   \begin{tabular}{lcccc}
      \tabularnewline \midrule \midrule
       & \multicolumn{2}{c}{Panel} & \multicolumn{2}{c}{LD} \\ 
      Model:                                  & (1)             & (2)             & (3)             & (4) \\  
      \midrule
      $DD_{it;0-29C}$                         & 0.0004$^{***}$  & 0.0003$^{***}$  & -0.0001         & -0.0002 \\   
                                              & (0.00008)       & (0.00008)       & (0.0003)        & (0.0003) \\   
      $DD_{it;>29C}$                          & -0.0057$^{***}$ & -0.0055$^{***}$ & -0.0053$^{***}$ & -0.0044$^{***}$ \\   
                                              & (0.0007)        & (0.0007)        & (0.0010)        & (0.0014) \\   
      $Prec_{it;\leq 42cm}$                   & 0.0117$^{***}$  & 0.0132$^{***}$  & 0.0515$^{**}$   & 0.0486$^{**}$ \\   
                                              & (0.0027)        & (0.0026)        & (0.0194)        & (0.0224) \\   
      $Prec_{it;> 42cm}$                      & -0.0007         & -0.0004         & 0.0036$^{**}$   & 0.0032$^{*}$ \\   
                                              & (0.0005)        & (0.0005)        & (0.0017)        & (0.0018) \\   

      \midrule
      Sample                                  & 1978-2002       & 1978-2002       & 1978-2002       & 1978-2002 \\  
      FE                                      & Cty, Yr         & Cty, Yr         & None            & None \\  
      Data                                    & BE (2016)       & Constructed     & BE (2016)       & Constructed \\ 
      Observations                            & 38,123          & 38,123          & 1,531           & 1,531 \\  
      \midrule \midrule
   \multicolumn{5}{p{0.62\linewidth}}{\footnotesize \textit{Notes}: 
Standard errors for estimates are clustered at the state level.}   
        \end{tabular}
\end{table}

\subsection{Detailed Simulation Tables}\label{app:sim_results}

\subsubsection{FE and LD Estimators}

The simulation results contained in this section are based on the simulation design in Table \ref{tab:simulation_design_1}. We conduct four variants of the simulation design, varying the lengths of the panel ($T=25$ and $T=40$) and varying the length of the climate normal (10- and 30-year). The different panel lengths span typical time horizons in the empirical literature (see Figure \ref{fig:ld_lit}).  

Table \ref{tab:simulation1_variantA} presents the simulation statistics for the FE ($\hat \beta_{FE}$) and LD estimators ($\hat \beta_{LD}$) across different of specifications of the climate normal, where $\theta_S=-0.07$ and $\theta_L=-0.02$.\footnote{As established in Proposition \ref{prop:FE} and \ref{prop:LD}, the weights do not depend on the values of $\theta_L$ and $\theta_S$.} The FE estimator exhibits small bias across all specifications. Consistent with Proposition \ref{prop:FE}, this negligible bias coincides with a small simulation mean of $\hat{\omega}_L^{FE}$, the estimated weight on $\theta_L$ in Proposition \ref{prop:FE}(i). For the LD estimator, the bias is substantive implying the estimator may not isolate climate variation from high-frequency weather variation. For instance, when $c^*_{it}$ is a 10-year normal and the specification uses $\tau = 10$, the simulation mean of $\hat \omega^{LD}_S$ is $0.3119$. The bias grows considerably as the moving-average horizon in $c^*_{it}$ becomes larger (see Table \ref{tab:simulation1_variantB}).

Figures \ref{fig:c10_ld_summary}, \ref{fig:c10_cov_rej_summary}, \ref{fig:c30_ld_summary},  and \ref{fig:c30_cov_rej_summary} present simulation results from the four variants across a grid of values for $\theta_L$, covering the no-adaptation case where $\theta_L=\theta_S$ to full adaptation where $\theta_L=0$.

Rejection probabilities of the adaptation test from long-different specifications in panels (c) and (d) of Figures \ref{fig:c10_cov_rej_summary} and \ref{fig:c30_cov_rej_summary} are compared with those obtained with the infeasible estimator. The infeasible estimator provides the rejection that could be attained if $c_{it}^*$ and $w_{it}^*$ were directly observed. Because $c_{it}^*$ is defined as a climate normal in this simulation, the infeasible estimator is implemented using fixed effects estimation that simultaneously estimates the short and long-run responses using $c_{it}^*$ and $w_{it}^*$ as regressors. Finite-sample power is substantially higher for smaller differences of $\theta_L-\theta_S$ with longer panels ($T=40$). 

Across all simulation variants, the coverage probability of $\theta_L-\theta_S$, the extent of adaptation, using the confidence intervals around $\hat{\beta}_{LD}-\hat{\beta}_{FE}$ is extremely low (Figures \ref{fig:c10_cov_rej_summary} and \ref{fig:c30_cov_rej_summary}). This is primarily driven by the bias discussed in Section \ref{sec:theory_results} and illustrated in panels (a) and (b) of Figures \ref{fig:c10_ld_summary} and  \ref{fig:c30_ld_summary}.  

\begin{table}[H]
\caption{Simulation results varying $\tau$ with $T=25$ and $\theta_L = -0.02$}
\centering
\label{tab:simulation1_variantA}
\begin{tabular}{ccccccc}
  \hline
\multicolumn{7}{c}{\textit{Long-Differences} ($\theta_L = -0.02$)} \\
Details & $\hat\beta_{LD}$ & $\hat\beta_{LD} - \theta_L$ & $\omega_S^{LD}$ & $\omega_L^{LD}$ & $\omega_S^{LD}\theta_S$ & $\omega_L^{LD}\theta_L$ \\
\hline
\multicolumn{7}{c}{\textit{Panel A. Climate as 10-year normal}} \\
 $\tau = 5$ & -0.0612 & -0.0412 & 0.8254 & 0.1746 & -0.0578 & -0.0035 \\ 
  $\tau = 10$ & -0.0356 & -0.0156 & 0.3119 & 0.6881 & -0.0218 & -0.0138 \\ 
   \multicolumn{7}{c}{\textit{Panel B. Climate as 30-year normal}} \\
$\tau = 5$ & -0.0632 & -0.0432 & 0.8642 & 0.1358 & -0.0605 & -0.0027 \\ 
  $\tau = 10$ & -0.0663 & -0.0463 & 0.9269 & 0.0731 & -0.0649 & -0.0015 \\ 
   \hline
\multicolumn{7}{c}{\textit{Fixed Effects} ($\theta_S = -0.07$)} \\
Details & $\hat\beta_{FE}$ & $\hat\beta_{FE} - \theta_S$ & $\omega_S^{FE}$ & $\omega_L^{FE}$ & $\omega_S^{FE}\theta_S$ & $\omega_L^{FE}\theta_L$ \\
\hline
\multicolumn{7}{c}{\textit{Panel A. Climate as 10-year normal}} \\
$\tau \in (5,10)$ & -0.0713 & -0.0013 & 1.0265 & -0.0265 & -0.0719 & 0.0005 \\ 
   \multicolumn{7}{c}{\textit{Panel B. Climate as 30-year normal}} \\
$\tau \in (5,10)$ & -0.0703 & -0.0003 & 1.0055 & -0.0055 & -0.0704 & 0.0001 \\ 
   \hline
\multicolumn{7}{p{0.7\linewidth}}{\footnotesize \textit{Notes:} This table presents the simulation means across 1,000 replications for $T = 25$.} \\
\end{tabular}
\end{table}

\begin{table}[H]
\centering
\caption{Simulation results varying $\tau$ with $T=40$ and $\theta_L = -0.02$} 
\label{tab:simulation1_variantB}
\begin{tabular}{ccccccc}
  \hline
\multicolumn{7}{c}{\textit{Long-Differences} ($\theta_L = -0.02$)} \\
Details & $\hat\beta_{LD}$ & $\hat\beta_{LD} - \theta_L$ & $\omega_S^{LD}$ & $\omega_L^{LD}$ & $\omega_S^{LD}\theta_S$ & $\omega_L^{LD}\theta_L$ \\
\hline
\multicolumn{7}{c}{\textit{Panel A. Climate as 10-year normal}} \\
 $\tau = 5$ & -0.0393 & -0.0193 & 0.3853 & 0.6147 & -0.0270 & -0.0123 \\ 
  $\tau = 10$ & -0.0217 & -0.0017 & 0.0339 & 0.9661 & -0.0024 & -0.0193 \\ 
  $\tau = 20$ & -0.0138 & 0.0062 & -0.1249 & 1.1249 & 0.0087 & -0.0225 \\ 
   \multicolumn{7}{c}{\textit{Panel B. Climate as 30-year normal}} \\
$\tau = 5$ & -0.0498 & -0.0298 & 0.5968 & 0.4032 & -0.0418 & -0.0081 \\ 
  $\tau = 10$ & -0.0439 & -0.0239 & 0.4785 & 0.5215 & -0.0335 & -0.0104 \\ 
  $\tau = 20$ & -0.0459 & -0.0259 & 0.5185 & 0.4815 & -0.0363 & -0.0096 \\ 
   \hline
\multicolumn{7}{c}{\textit{Fixed Effects} ($\theta_S = -0.07$)} \\
Details & $\hat\beta_{FE}$ & $\hat\beta_{FE} - \theta_S$ & $\omega_S^{FE}$ & $\omega_L^{FE}$ & $\omega_S^{FE}\theta_S$ & $\omega_L^{FE}\theta_L$ \\
\hline
\multicolumn{7}{c}{\textit{Panel A. Climate as 10-year normal}} \\
$\tau \in (5,10,20)$ & -0.0703 & -0.0003 & 1.0053 & -0.0053 & -0.0704 & 0.0001 \\ 
   \multicolumn{7}{c}{\textit{Panel B. Climate as 30-year normal}} \\
$\tau \in (5,10,20)$ & -0.0692 & 0.0008 & 0.9837 & 0.0163 & -0.0689 & -0.0003 \\ 
   \hline
\multicolumn{7}{p{0.7\linewidth}}{\footnotesize \textit{Notes:} This table presents the simulation means across 1,000 replications for $T = 25$.} \\
\end{tabular}
\end{table}

\newpage

\begin{figure}
\centering
\begin{subfigure}{0.48\textwidth}
    \centering
\includegraphics[width=\textwidth]{Figures/diff_ld_10_T25_july22.pdf}
    \caption{Estimates for T = 25}
\end{subfigure}
\begin{subfigure}{0.48\textwidth}
    \centering
    \includegraphics[width=\textwidth]{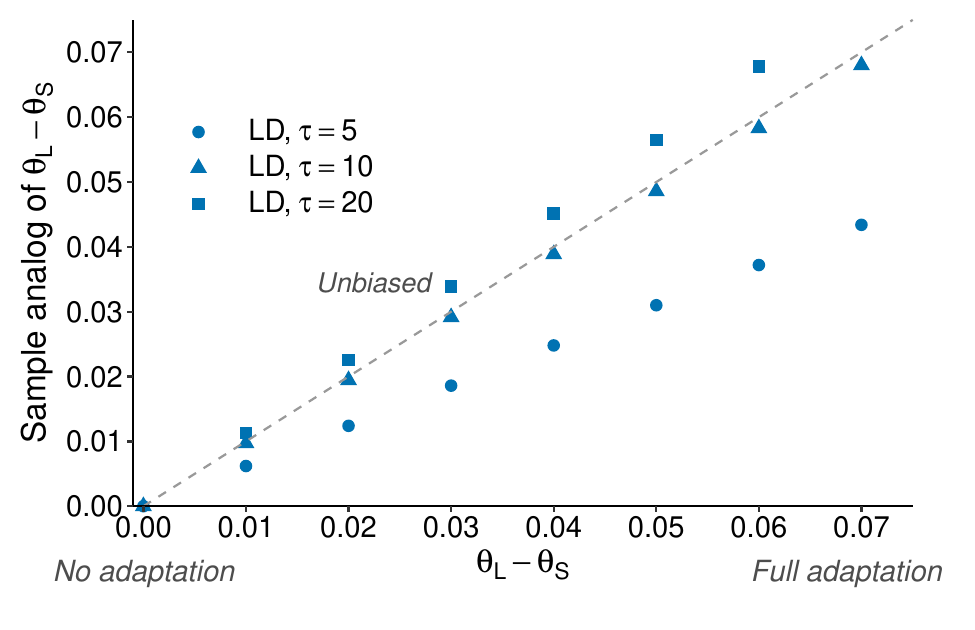}
       \caption{Estimates for T = 40}
\end{subfigure}
\begin{subfigure}{0.48\textwidth}
    \centering
    \includegraphics[width=\textwidth]{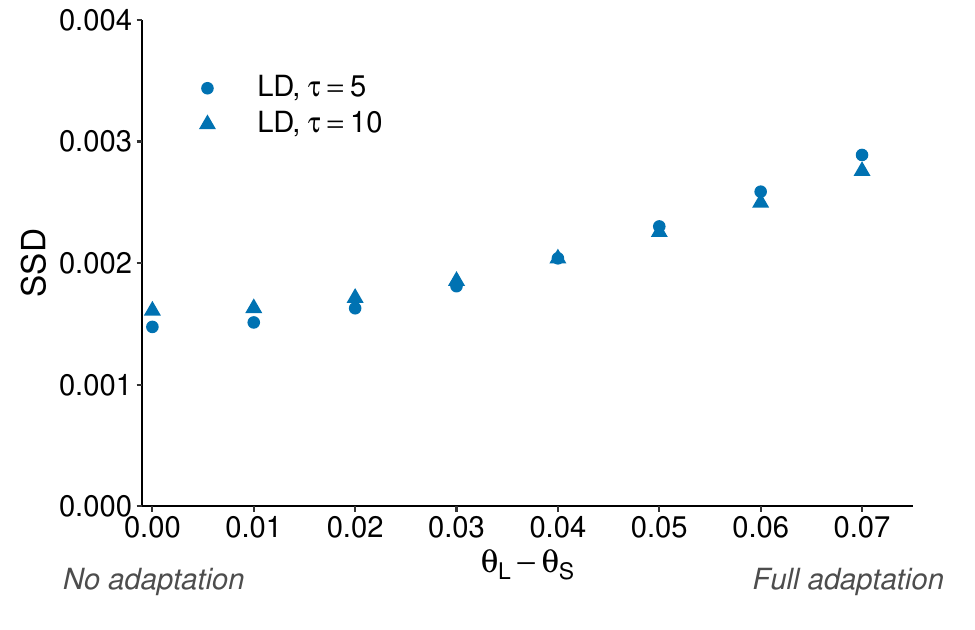}
    \caption{Simulation SD for T = 25}
\end{subfigure}
\begin{subfigure}{0.48\textwidth}
    \centering
    \includegraphics[width=\textwidth]{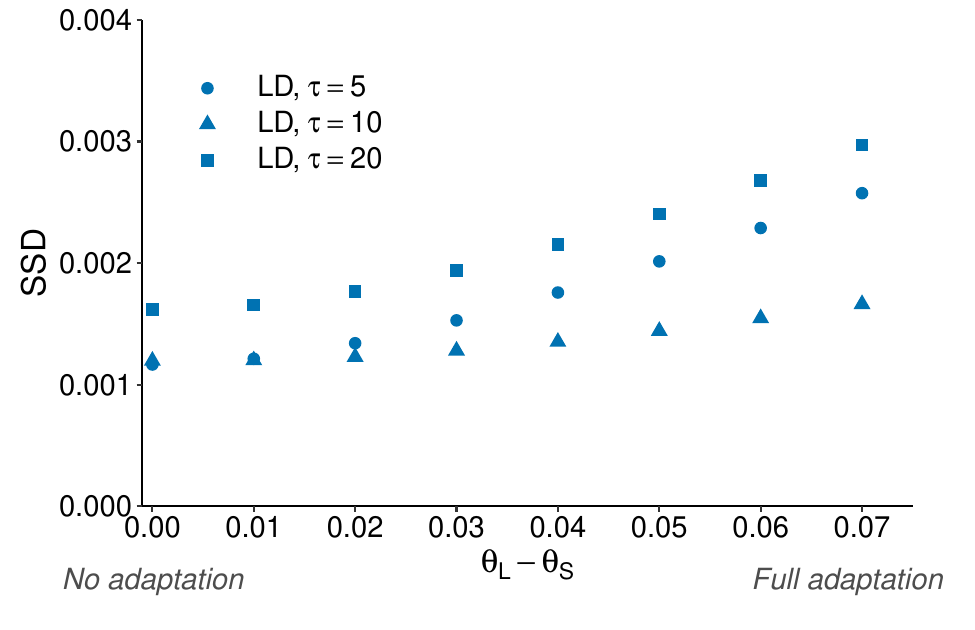}
     \caption{Simulation SD for T = 40}
\end{subfigure}
\begin{subfigure}{0.48\textwidth}
    \centering
    \includegraphics[width=\textwidth]{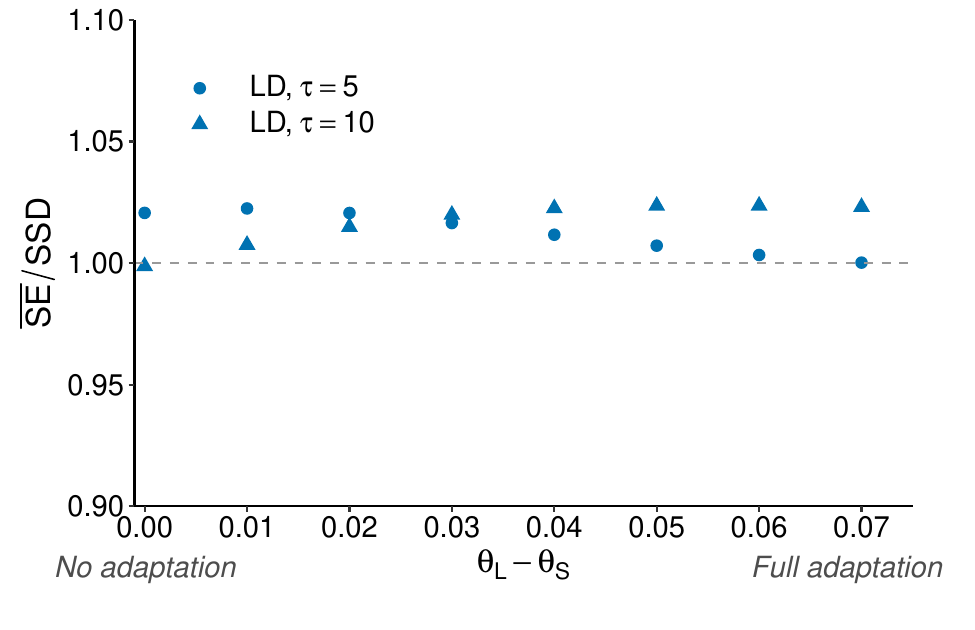}
    \caption{$\overline{SE}/SSD$ ratio for T = 25}
\end{subfigure}
\begin{subfigure}{0.48\textwidth}
    \centering
    \includegraphics[width=\textwidth]{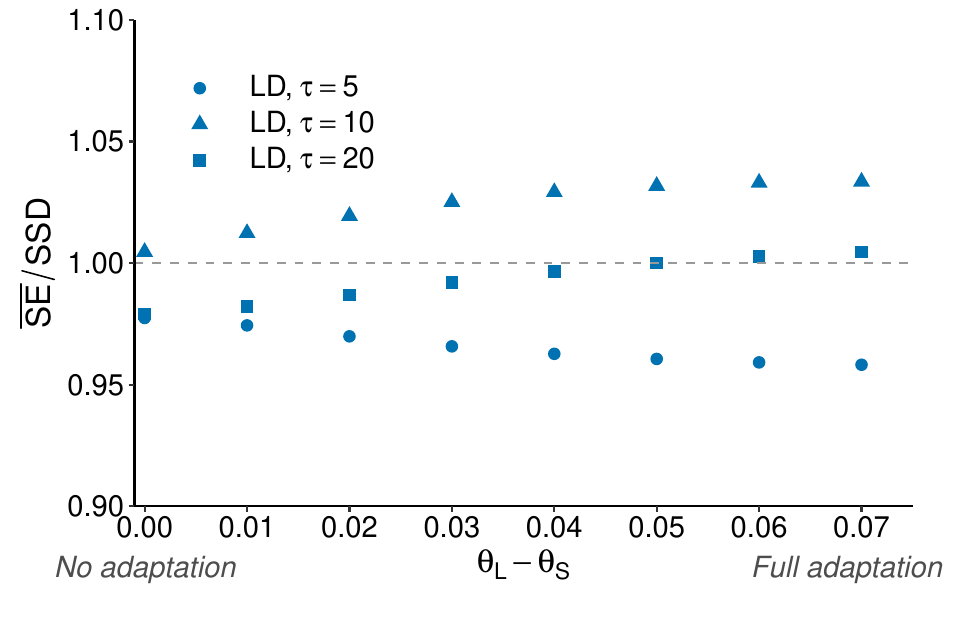}
     \caption{$\overline{SE}/SSD$ ratio for T = 40}
\end{subfigure}
   \caption{LD--FE simulation results for $c_{it}^*$ as 10-year normal: Summary statistics}
     \label{fig:c10_ld_summary}
\caption*{\scriptsize \textit{Notes:} Panels (a) and (b) plot $\bar{\hat \beta}_{LD}-\bar{\hat\beta}_{FE}$ against $\theta_{L}-\theta_{S}$. Panels (c) and (d) plot the simulation standard deviation of $\bar{\hat \beta}_{LD}-\bar{\hat\beta}_{FE}$. Panels (e) and (f) plot the ratio of the simulation mean of the standard error of $\hat{\beta}_{LD}-\hat{\beta}_{FE}$ to its simulation standard deviation, where the cluster-robust bootstrap standard error is computed using 400 bootstrap replications with counties treated as clusters. }
 
\end{figure}

\begin{figure}
\centering
\begin{subfigure}{0.48\textwidth}
    \centering
    \includegraphics[width=\textwidth]{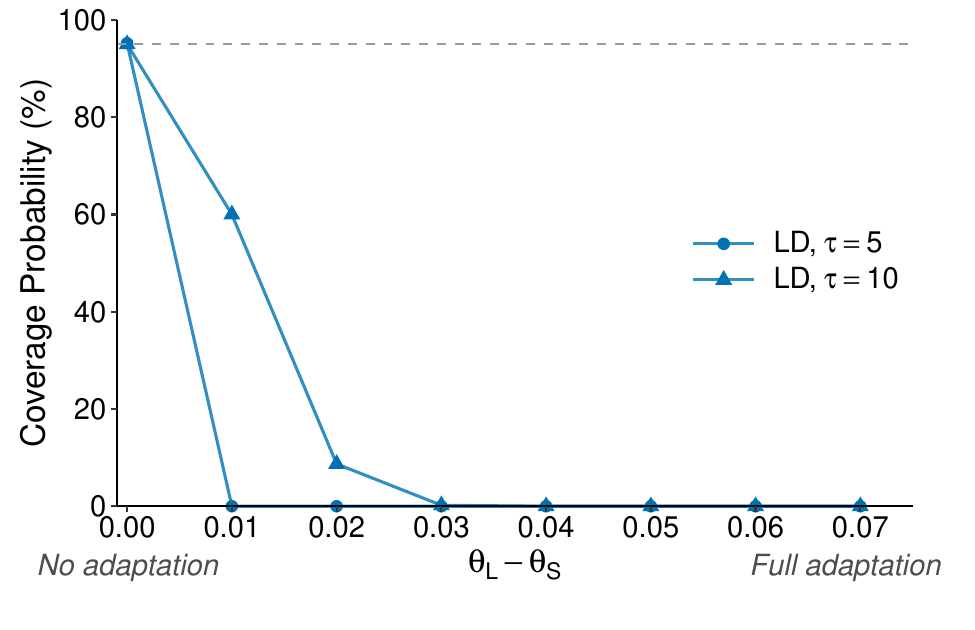}
    \caption{Coverage probabilities for T = 25}
\end{subfigure}
\begin{subfigure}{0.48\textwidth}
    \centering
    \includegraphics[width=\textwidth]{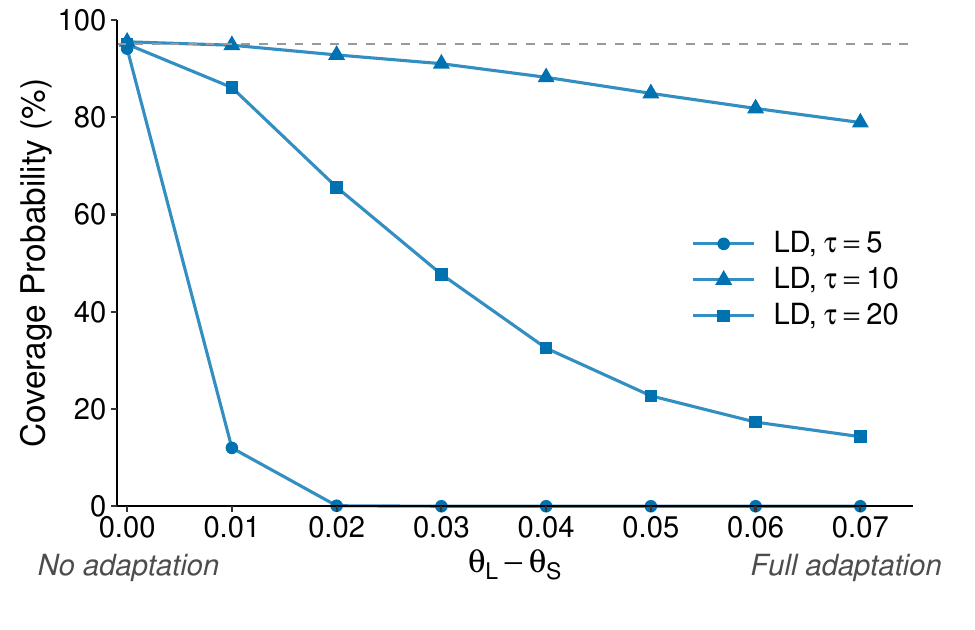}
     \caption{Coverage probabilities for T = 40}
\end{subfigure}
\begin{subfigure}{0.48\textwidth}
    \centering
    \includegraphics[width=\textwidth]{Figures/rejections_ld_infeasible_10_T25_july22.pdf}
    \caption{Rejection probabilities for T = 25}
\end{subfigure}
\begin{subfigure}{0.48\textwidth}
    \centering
    \includegraphics[width=\textwidth]{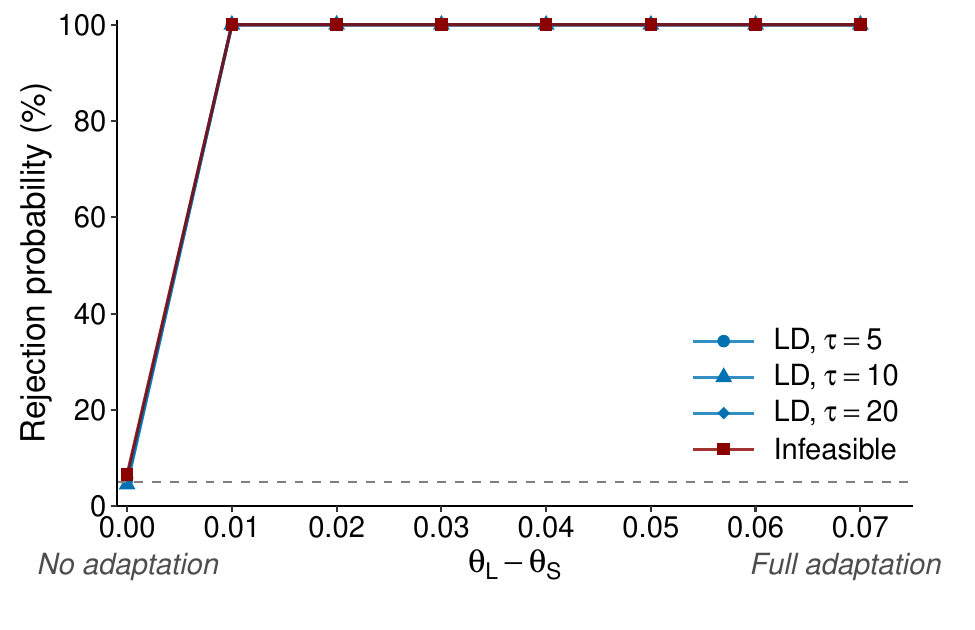}
     \caption{Rejection probabilities for T = 40}
\end{subfigure}
   \caption{LD--FE simulation results for $c^*_{it}$ as 10-year normal: Coverage and rejection probability}
\caption*{\scriptsize \textit{Notes:} Panels (a) and (b) plot the simulation analogue of $P(\theta_L-\theta_S\in \text{95 \% CI})$, where 95\% CI is the confidence interval around $\hat \beta_{FE} - \hat \beta_{LD}$ using bootstrapped standard errors ($B=400$). Panels (c) and (d) plot the rejection probability curves at the $\alpha = 0.05$ significance level for the null hypothesis $H_0: \beta_{LD} = \beta_{FE}$ using the LD and FE estimator. The null hypothesis of the infeasible test is $H_0:\theta_L=\theta_S$, which are estimated using $c_{it}^*$ and $w_{it}^*$ as regressors in a fixed-effects model. Simulation rejection probabilities are computed across 1,000 replications for  for a grid of values of $\theta_L$ over $\{-0.07,-0.06,0\}$ with $\theta_S=-0.07$.}
   \label{fig:c10_cov_rej_summary}
\end{figure}

\begin{figure}
\centering
\begin{subfigure}{0.48\textwidth}
    \centering
\includegraphics[width=\textwidth]{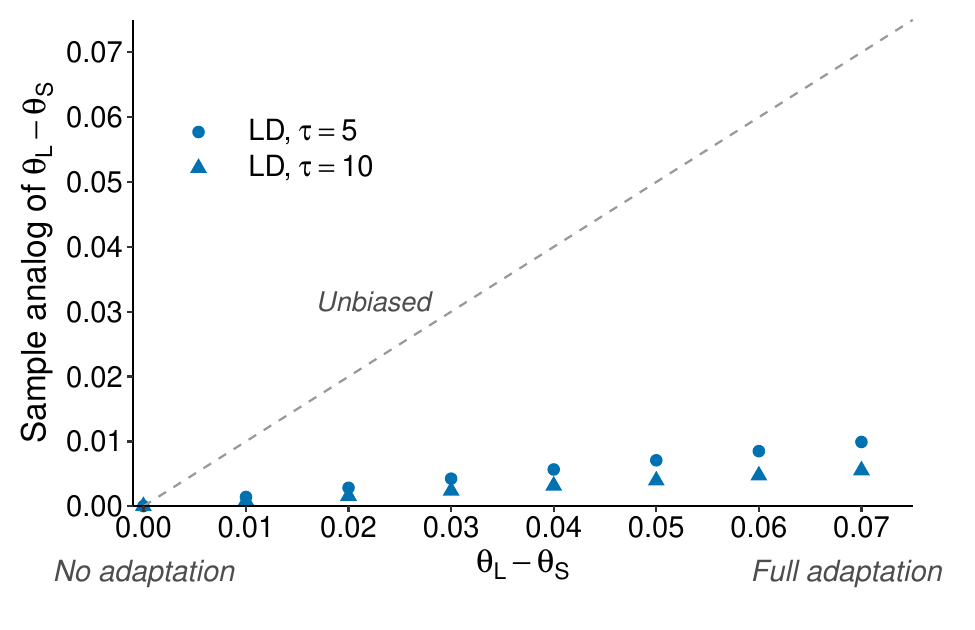}
    \caption{Estimates for T = 25}
\end{subfigure}
\begin{subfigure}{0.48\textwidth}
    \centering
    \includegraphics[width=\textwidth]{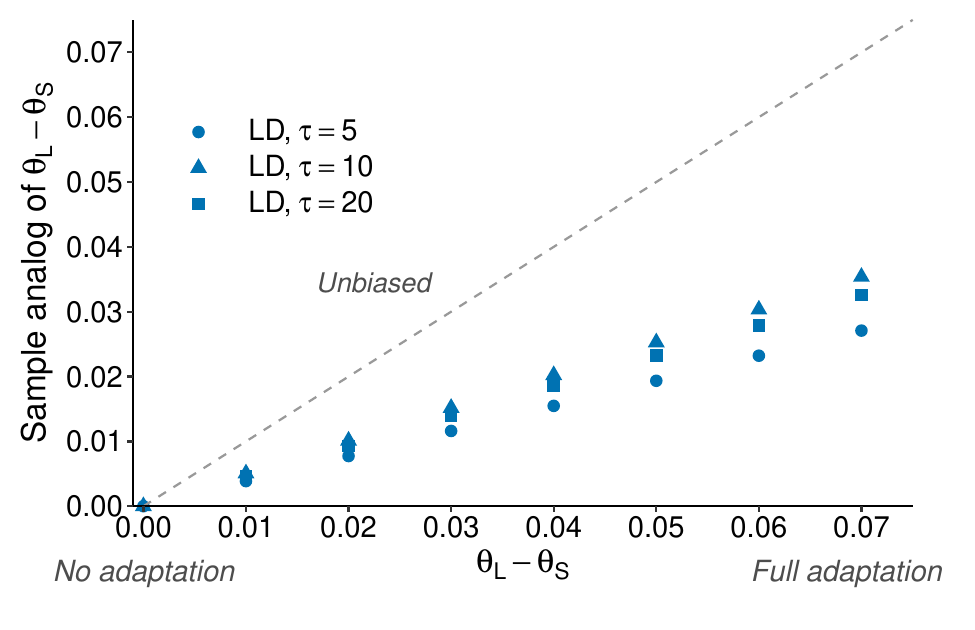}
       \caption{Estimates for T = 40}
\end{subfigure}

\begin{subfigure}{0.48\textwidth}
    \centering
    \includegraphics[width=\textwidth]{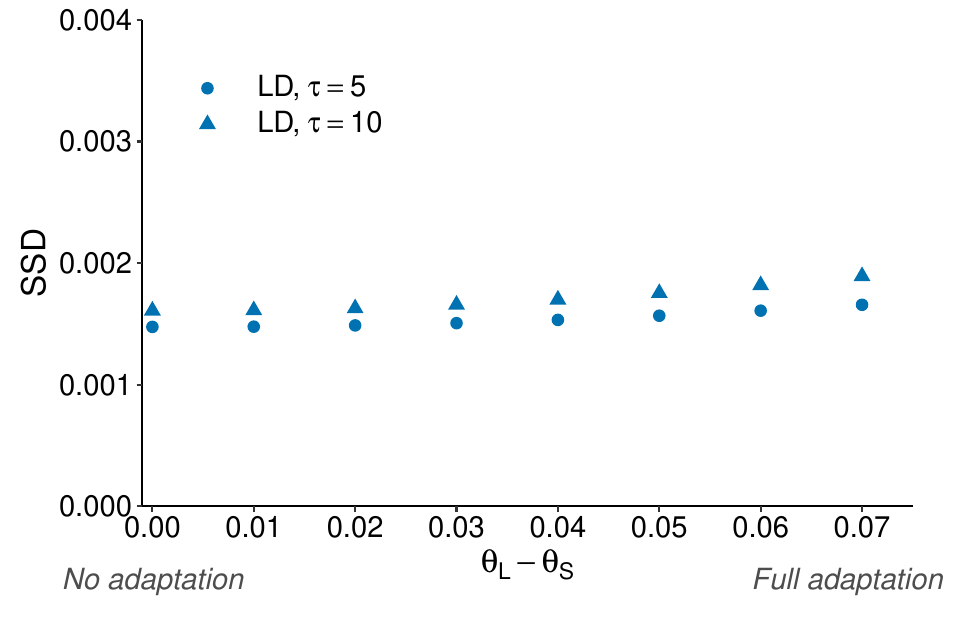}
    \caption{Simulation SD for T = 25}
\end{subfigure}
\begin{subfigure}{0.48\textwidth}
    \centering
    \includegraphics[width=\textwidth]{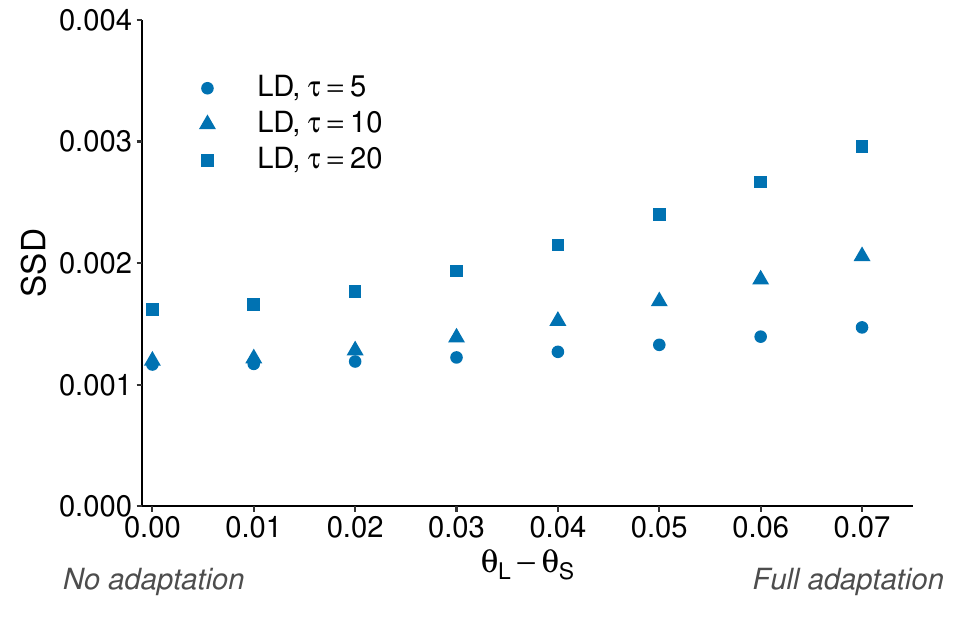}
     \caption{Simulation SD for T = 40}
\end{subfigure}
\begin{subfigure}{0.48\textwidth}
    \centering
    \includegraphics[width=\textwidth]{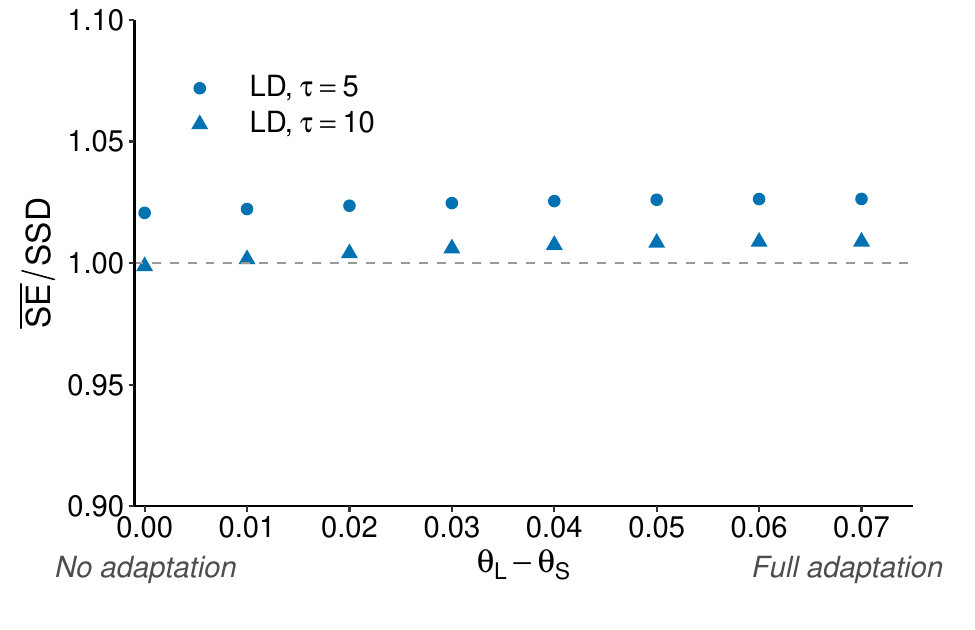}
    \caption{$\overline{SE}/SSD$ ratio for T = 25}
\end{subfigure}
\begin{subfigure}{0.48\textwidth}
    \centering
    \includegraphics[width=\textwidth]{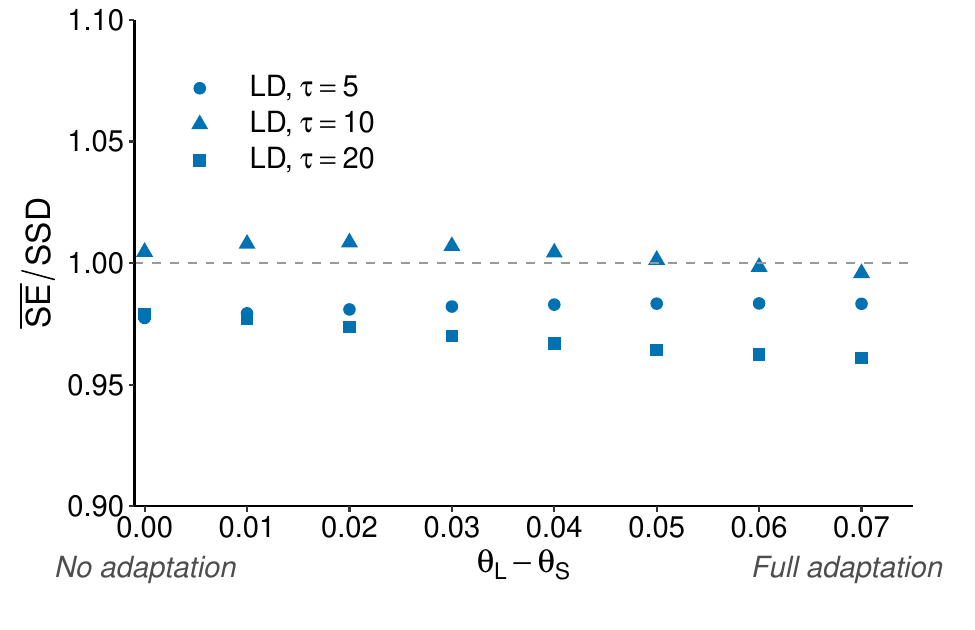}
     \caption{$\overline{SE}/SSD$ ratio for T = 40}
\end{subfigure}
   \caption{LD--FE simulation results for $c_{it}^*$ as 30-year normal: Summary statistics}
\caption*{\scriptsize \textit{Notes:} Panels (a) and (b) plot $\bar{\hat \beta}_{LD}-\bar{\hat\beta}_{FE}$ against $\theta_{L}-\theta_{S}$. Panels (c) and (d) plot the simulation standard deviation of $\bar{\hat \beta}_{LD}-\bar{\hat\beta}_{FE}$. Panels (e) and (f) plot the ratio of the simulation mean of the standard error of $\hat{\beta}_{LD}-\hat{\beta}_{FE}$ to its simulation standard deviation, where the cluster-robust bootstrap standard error is computed using 400 bootstrap replications with counties treated as clusters. }
   \label{fig:c30_ld_summary}
\end{figure}

\begin{figure}
\centering
\begin{subfigure}{0.48\textwidth}
    \centering
    \includegraphics[width=\textwidth]{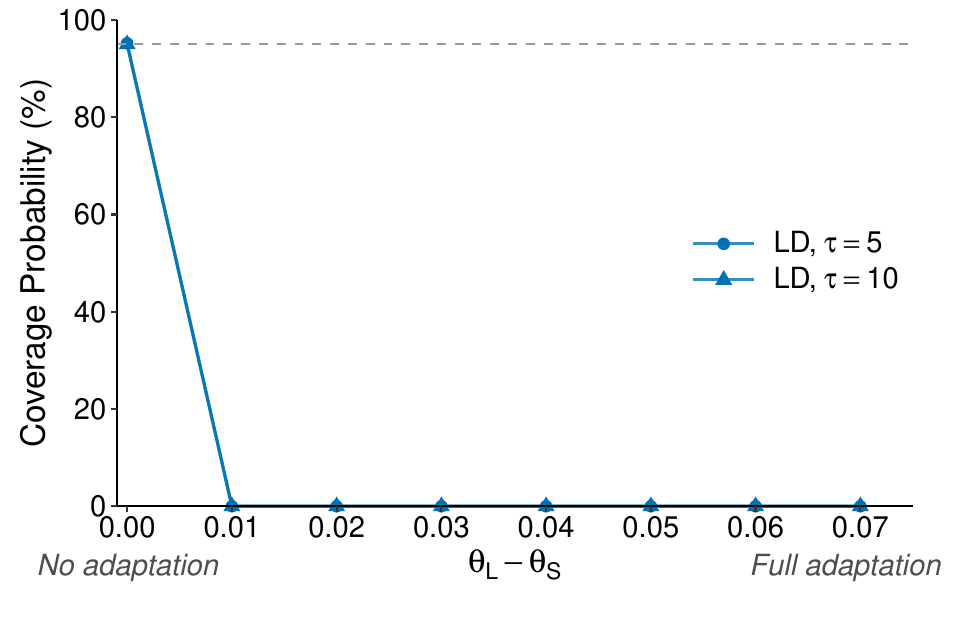}
    \caption{Coverage probabilities for T = 25}
\end{subfigure}
\begin{subfigure}{0.48\textwidth}
    \centering
    \includegraphics[width=\textwidth]{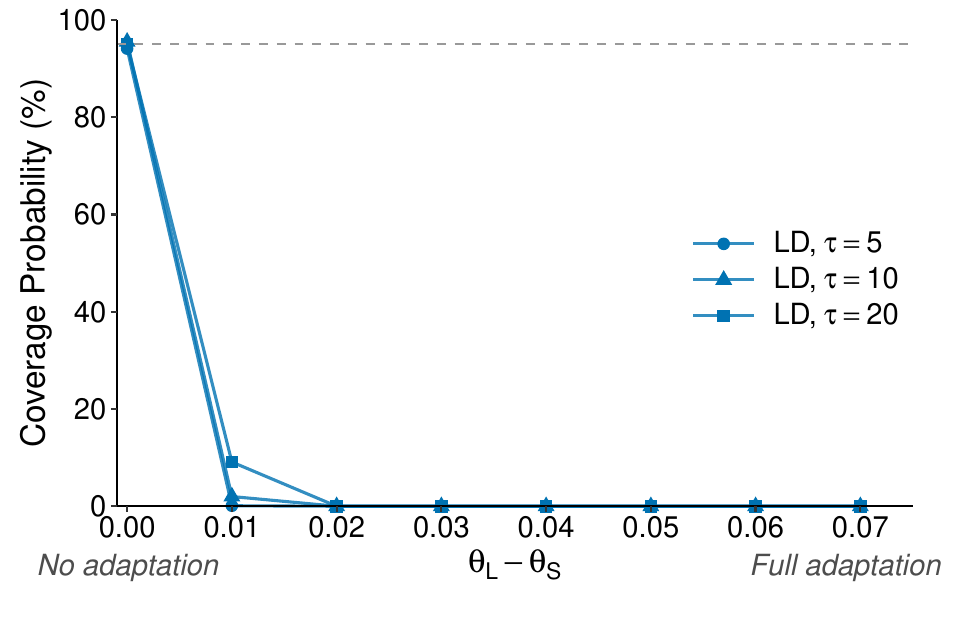}
     \caption{Coverage probabilities for T = 40}
\end{subfigure}
\begin{subfigure}{0.48\textwidth}
    \centering
    \includegraphics[width=\textwidth]{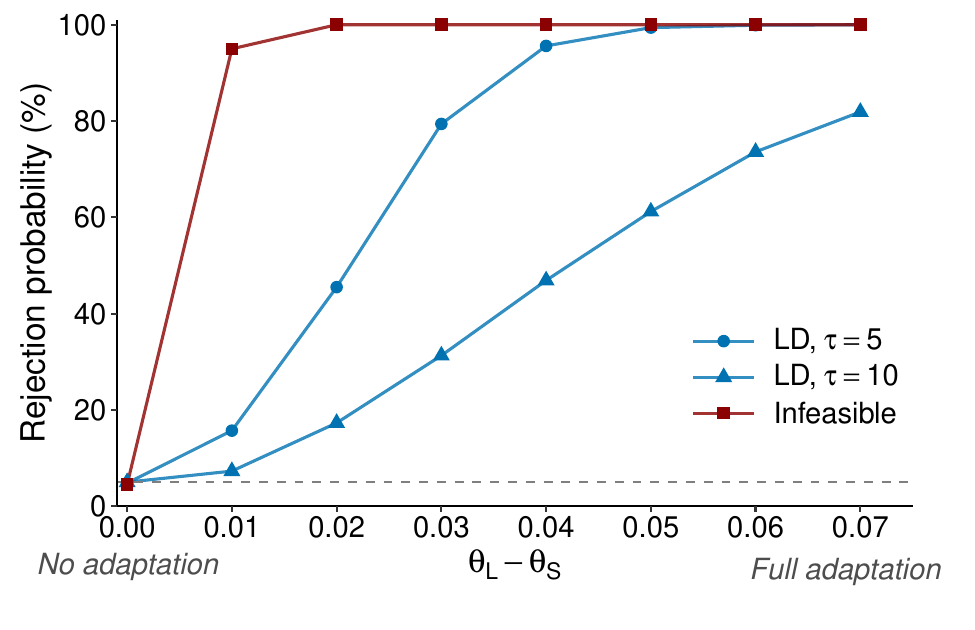}
    \caption{Rejection probabilities for T = 25}
\end{subfigure}
\begin{subfigure}{0.48\textwidth}
    \centering
    \includegraphics[width=\textwidth]{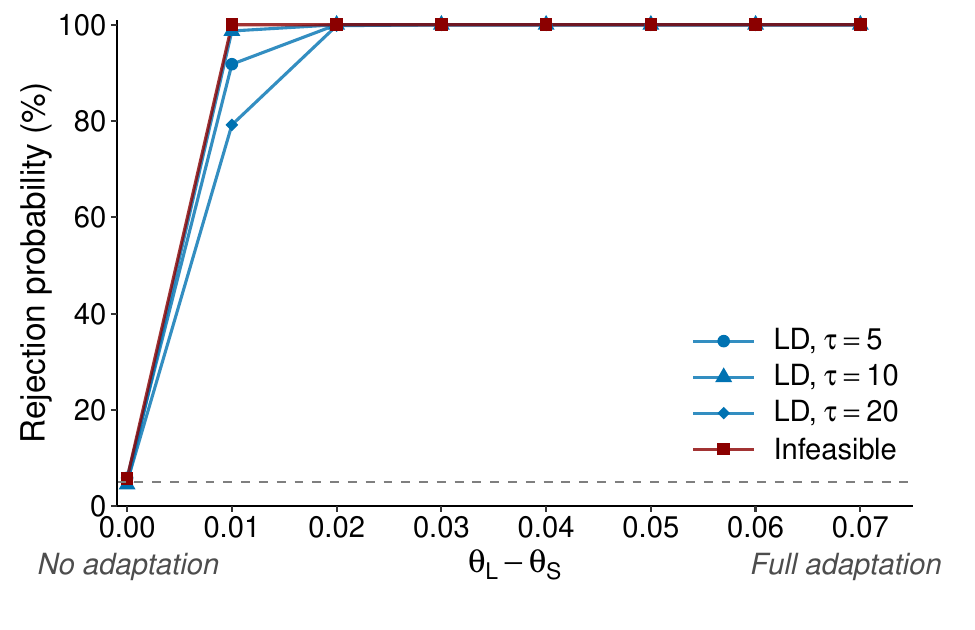}
     \caption{Rejection probabilities for T = 40}
\end{subfigure}
   \caption{LD--FE simulation results for $c^*_{it}$ as 30-year normal: Coverage and rejection probability}\label{fig:sim_rej_cov_30-year}
\caption*{\scriptsize \textit{Notes:} Panels (a) and (b) plot the simulation analogue of $P(\theta_L-\theta_S\in \text{95 \% CI})$, where 95\% CI is the confidence interval around $\hat \beta_{FE} - \hat \beta_{LD}$ using bootstrapped standard errors ($B=400$). Panels (c) and (d) plot the rejection probability curves at the $\alpha = 0.05$ significance level for the null hypothesis $H_0: \beta_{LD} = \beta_{FE}$ using the LD and FE estimator. The null hypothesis of the infeasible test is $H_0:\theta_L=\theta_S$, which are estimated using $c_{it}^*$ and $w_{it}^*$ as regressors in a fixed-effects model. Simulation rejection probabilities are computed across 1,000 replications for  for a grid of values of $\theta_L$ over $\{-0.07,-0.06,0\}$ with $\theta_S=-0.07$.}
   \label{fig:c30_cov_rej_summary}
\end{figure}

\begin{figure}
\centering
\begin{subfigure}{0.48\textwidth}
    \centering
    \includegraphics[width=\textwidth]{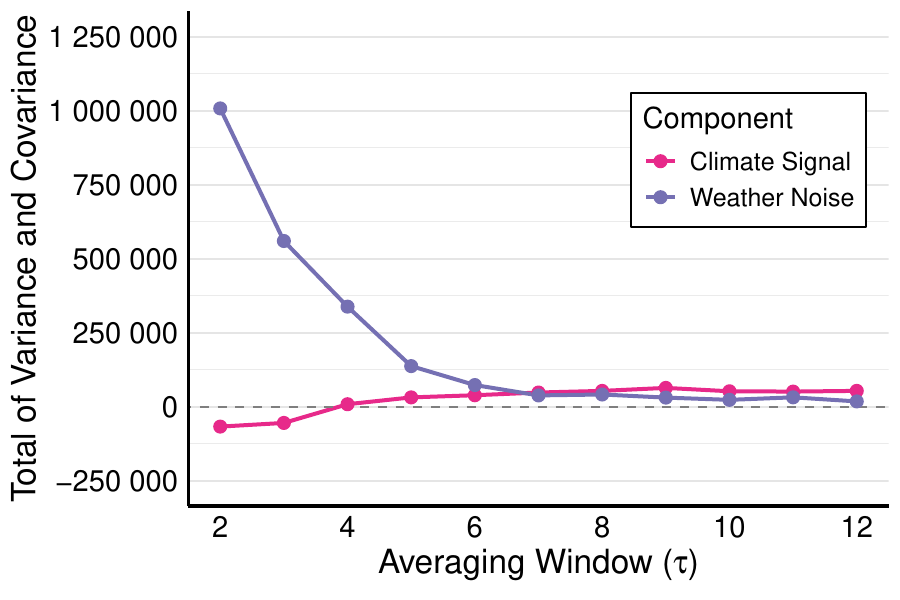}
    \caption{Climate as 10-year normal}
\end{subfigure}
\begin{subfigure}{0.48\textwidth}
    \centering
    \includegraphics[width=\textwidth]{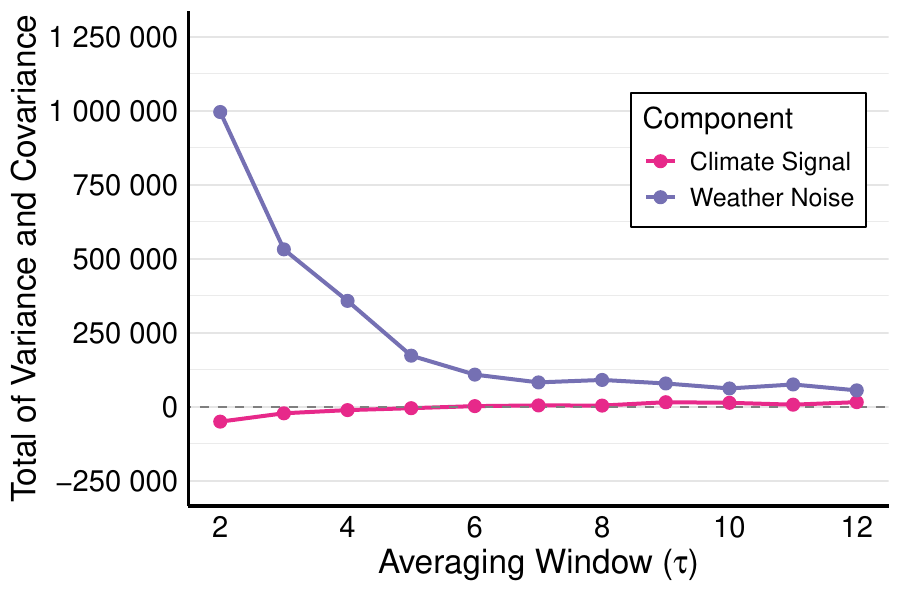}
    \caption{Climate as 20-year normal}
\end{subfigure}
\begin{subfigure}{0.48\textwidth}
    \centering
    \includegraphics[width=\textwidth]{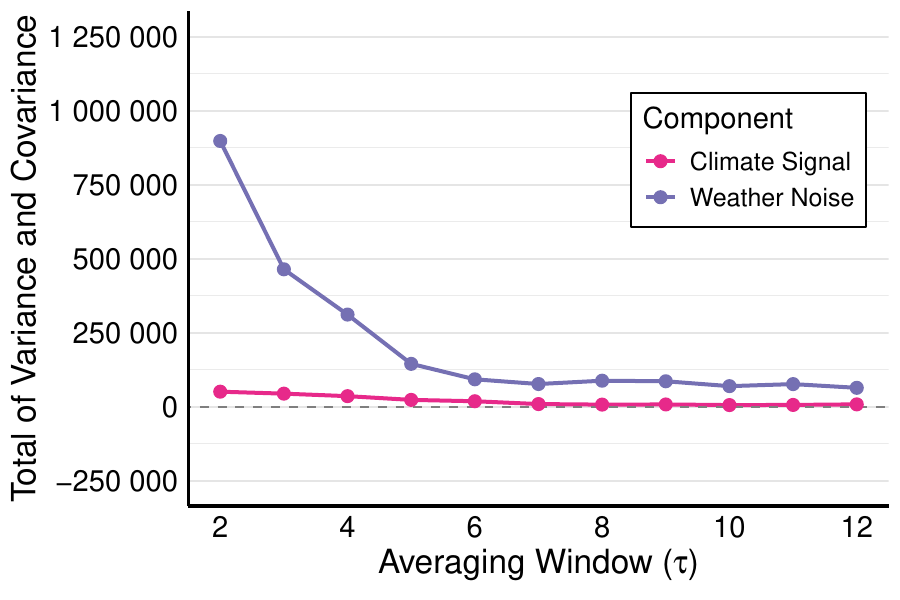}
    \caption{Climate as 30-year normal}
\end{subfigure}

   \caption{Numerator (climate signal) and denominator (weather noise) of the signal-to-noise ratio for values of $\tau$ for $T=25$}
\end{figure}

\begin{figure}
\centering
\begin{subfigure}{0.48\textwidth}
    \centering
    \includegraphics[width=\textwidth]{Figures/var_covar_T25_clim10.pdf}
    \caption{Climate as 10-year normal}
\end{subfigure}
\begin{subfigure}{0.48\textwidth}
    \centering
    \includegraphics[width=\textwidth]{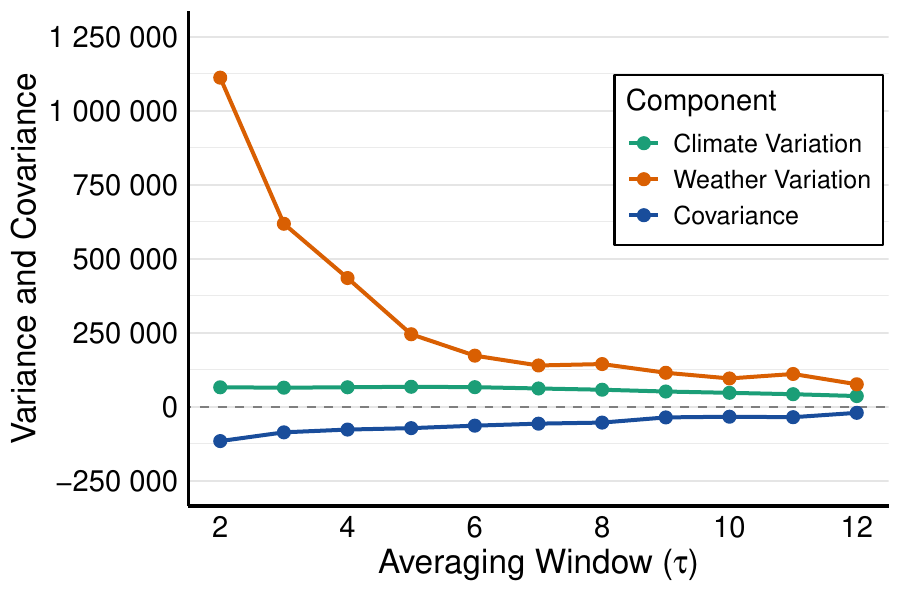}
    \caption{Climate as 20-year normal}
\end{subfigure}
\begin{subfigure}{0.48\textwidth}
    \centering
    \includegraphics[width=\textwidth]{Figures/var_covar_T25_clim30.pdf}
    \caption{Climate as 30-year normal}
\end{subfigure}
   \caption{Climate variation, weather variation, and covariance terms used to construct $SNR_c$ and $\omega_S^{LD}$ across values of $\tau$ for $T=25$}
\end{figure}

\end{document}